\documentclass[11pt]{article}

\usepackage{jeffstyle}
\usepackage{cite}
\usepackage{pbox}
\usepackage{amsmath}
\usepackage{hyperref}
\usepackage{url}
\usepackage{algorithm,algorithmic,graphicx,tikz}
\usepackage{amsopn,amssymb,amsthm,amsmath}

\newcommand{\eqdef}{\mathrel{\mathop=}:}

\newcommand{\R}{\mathbb{R}}

\newlength{\figsize} 
\usepackage{subfigure,color}

\usepackage{color}

\newcounter{property_counter}
\newtheorem{property}[property_counter]{Property}

\title{Frequent Directions : Simple and Deterministic Matrix Sketching}
\author{
Mina Ghashami\\ University of Utah \\ \texttt{ghashami@cs.utah.edu}
\and
Edo Liberty\\ Yahoo Labs\\ \texttt{edo.liberty@yahoo.com}
\and
Jeff M. Phillips\thanks{Thanks to support by NSF CCF-1350888, IIS-1251019, and ACI-1443046.}\\ University of Utah \\ \texttt{jeffp@cs.utah.edu}
\and
David P. Woodruff\thanks{Supported by the XDATA program of DARPA administered through Air Force Research Laboratory contract FA8750-12-C0323.}\\ IBM Research--Almaden \\ \texttt{dpwoodru@us.ibm.com}
}

\begin{document}
\date\nonumber
\maketitle

\begin{abstract}
We describe a new algorithm called Frequent Directions for deterministic matrix sketching in the row-updates model.  
The algorithm is presented an arbitrary input matrix $A \in \R^{n \times d}$ one row at a time. It performed $O(d\ell)$ operations per row and maintains a sketch matrix $B \in \R^{\ell \times d}$ such that for any $k < \ell$
\[
\|A^TA - B^TB \|_2 \leq \|A - A_k\|_F^2 / (\ell-k) \mbox{\;\;\;\;and\;\;\;\;} \|A - \pi_{B_k}(A)\|_F^2 \leq \big(1 + \frac{k}{\ell-k}\big) \|A-A_k\|_F^2 \ .
\]
Here, $A_k$ stands for the minimizer of $\|A - A_k\|_F$ over all rank $k$ matrices (similarly $B_k$) and $\pi_{B_k}(A)$ is the rank $k$ matrix resulting from projecting $A$ on the row span of $B_k$.\footnote{The matrix $A_0$ is defined to be an all zeros matrix of the appropriate dimensions.}

We show both of these bounds are the best possible for the space allowed.  
The summary is mergeable, and hence trivially parallelizable.  
Moreover, Frequent Directions outperforms exemplar implementations of existing streaming algorithms in 
the space-error tradeoff.\footnote{This paper combines, extends and simplifies the results in \cite{Lib12}\cite{GhashamiP14} and \cite{Woo14}.}
\end{abstract}

\section{Introduction}
\label{sec:intro}

The data streaming paradigm~\cite{muthukrishnan05:_book, gibbons1999synopsis} considers computation on a large data set $A$ where data items arrive in arbitrary order, are processed, and then never seen again.  
It also enforces that only a small amount of memory is available at any given time.  
This small space constraint is critical when the full data set cannot fit in memory or disk.  
Typically, the amount of space required is traded off with the accuracy of the computation on $A$.  
Usually the computation results in some summary $S(A)$ of $A$, and this trade-off determines how accurate one can be with the available space resources.  

Modern large data sets are often viewed as large matrices. For example, textual data in the bag-of-words model is represented by a matrix whose rows correspond to documents. In large scale image analysis, each row in the matrix corresponds to one image and contains either pixel values or other derived feature values. Other large scale machine learning systems generate such matrices by converting each example into a list of numeric features. Low rank approximations for such matrices are used in common data mining tasks such as Principal Component Analysis (PCA), Latent Semantic Indexing (LSI), and k-means clustering. Regardless of the data source, the optimal low rank approximation for any matrix is obtained by its truncated Singular Value Decompositions (SVD).

In large matrices as above, one processor (and memory) is often incapable of handling all of the dataset $A$ in a feasible amount of time.  Even reading a terabyte of data on a single processor can take many hours. Thus this computation is often spread among some set of machines. This renders standard SVD algorithms infeasible. Given a very large matrix $A$, a common approach is to compute  in the streaming paradigm a sketch matrix $B$ that is significantly smaller than the original. A good sketch matrix $B$ is such that $A \approx B$ or $A^T A \approx B^T B$ and so computations can be performed on $B$ rather than on $A$ without much loss in precision.

Prior to this work, there are three main matrix sketching approaches, presented here in an arbitrary order. 
The first generates a sparser version of the matrix. Sparser matrices are stored more efficiently and can be multiplied faster by other matrices~\cite{arora2006fast,achlioptas2001fast,drineas2011note}. 
The second approach is to randomly combine matrix rows~\cite{papadimitriou1998latent,vempala2004random,sarlos2006improved,liberty2007randomized}. The proofs for these rely on subspace embedding techniques and strong concentration of measure phenomena. The above methods will be collectively referred to as \emph{random-projection} in the experimental section. 
A recent result along these lines~\cite{clarkson2013low}, gives simple and efficient subspace embeddings that can be applied in time $O(nnz(A))$ for any matrix $A$. We will refer to this result as \emph{hashing} in the experimental section. While our algorithm requires more computation than hashing, it will produce more accurate estimates given a fixed sketch size. 
The third sketching approach is to find a small subset of matrix rows (or columns) that approximate the entire matrix. This problem is known as the `Column Subset Selection Problem' and has been thoroughly investigated~\cite{frieze2004fast,drineas2003pass,boutsidis2009improved,deshpande2006adaptive,drineas2011faster,boutsidis2011near}. Recent results offer algorithms with almost matching lower bounds~\cite{cw09,boutsidis2011near,deshpande2006adaptive}. 
A simple streaming solution to the `Column Subset Selection Problem' is obtained by sampling rows from the input matrix with probability proportional to their squared $\ell_2$ norm. Despite this algorithm's apparent simplicity, providing tight bounds for its performance required over a decade of research 
~\cite{frieze2004fast,ahlswede2002strong,drineas2003pass,rudelson2007sampling,vershynin2009note,oliveira2010sums,drineas2011faster}. We will refer to this algorithm as \emph{sampling}. Algorithms such as CUR utilize the leverage scores of the rows~\cite{drineas2008relative} and not their squared $\ell_2$ norms. The discussion on matrix leverage scores goes beyond the scope of this paper, see~\cite{drineas2011faster} for more information and references.

In this paper, we propose a fourth approach: \emph{frequent directions}.  It is deterministic and draws on the similarity between the matrix sketching problem and the item frequency estimation problem.  
We provide additive and relative error bounds for it; 
we show how to merge summaries computed in parallel on disjoint subsets of data; 
we show it achieves the optimal tradeoff between space and accuracy, up to constant factors, for any row-update based summary; 
and we empirically demonstrate that it outperforms exemplars from all of the above described approaches.

\subsection{Notations and preliminaries}
Throughout this manuscript, for vectors, $\|\cdot\|$ will denote the Euclidian norm of a vectors $\|x\| = \sqrt{\sum_i{x_i^2}}$.
For matrices $\|\cdot\|$ will denote the operator (or spectral) norm $\|A\| = \sup_{\|x\|=1} \|Ax\|$.
Unless otherwise stated, vectors are assumed to be column vectors. The notation $a_i$ denotes the $i$th row of the matrix $A$.
The Frobenius norm of a matrix $A$ is defined as $\|A\|_F = \sqrt{\sum_{i=1} \|a_i\|^2}$.  
$I_\ell$ refers to the $\ell \times \ell$ identity matrix.

The singular value decomposition of matrix $B \in \mathbb{R}^{\ell \times d}$ is denoted by $[U,\Sigma ,V] = \svd(B)$. 
If $\ell \le d$ it guarantees that $B = U\Sigma V^T$, $U^TU = I_\ell$, $V^TV = I_d$, $U\in \R^{\ell \times \ell}$, $V\in \R^{d \times \ell}$,
and $\Sigma \in \R^{\ell \times \ell}$ is a non-negative diagonal matrix such that $\Sigma_{1,1}\ge \Sigma_{2,2} \ge \ldots \ge \Sigma_{\ell,\ell} \ge 0$. 
It is convenient to denote by $U_k$, and $V_k$ the matrices containing the first $k$ columns of $U$ and $V$ and $\Sigma_k \in \R^{k \times k}$ the top left $k \times k$ block of $\Sigma$.
The matrix $A_k = U_k \Sigma_k V_k^T$ is the best rank $k$ approximation of $A$ in the sense that $A_k = {\arg \min}_{C : \rank(C) \leq k} \|A - C\|_{2,F}$.  
Finally we denote by $\pi_B(A)$ the projection of the rows $A$ on the span of the rows of $B$. In other words, $\pi_B(A) = A B^\dagger B$ where $(\cdot)^\dagger$ indicates taking the Moore-Penrose psuedoinverse.
Alternatively, setting $[U,\Sigma ,V] = \svd(B)$, we also have $\pi_B(A) = AV^TV$. 
Finally, we denote $\pi_B^k(A) = AV_k^TV_k$, the right projection of $A$ on the top $k$ right singular vectors of $B$.

\subsection{Item Frequency Approximation}
Our algorithm for sketching matrices is an extension of a well known algorithm for approximating item frequencies in streams.
The following section shortly overviews the frequency approximation problem.
Here, a stream $A = \{a_1,\cdots, a_n\}$ has $n$ elements where each $a_i \in [d]$.  Let $f_j = |\{a_i \in A \mid a_i = j\}|$ be the frequency of item $j$ and stands for number of times item $j$ appears in the stream.
It is trivial to produce all item frequencies using $O(d)$ space simply by keeping a counter for each item. Although this method computes exact frequencies, it uses space linear to the size of domain which might be huge. Therefore, we are interested in using less space and producing approximate frequencies $\hat{f}_j$.

This problem received an incredibly simple and elegant solution by Misra and Gries~\cite{mg-fre-82}. 
Their algorithm~\cite{mg-fre-82} employes a map of $\ell < d$ items to $\ell$ counters. 
It maintains the invariant that at least one of the items is mapped to counter of value zero.
The algorithm counts items in the trivial way. If it encounters an item for which is has a counter, that counter is increased. 
Else, it replaces one of the items mapping to zero value with the new item (setting the counter to one).
This is continued until the invariant is violated, that is, $\ell$ items map to counters of value at least $1$.
At this point, all counts are decreased by the same amount until at least one item maps to a zero value.
The final values in the map give approximate frequencies  $\hat f_j$ such that $0 \leq f_j - \hat f_j \leq n/\ell$ for all $j \in [d]$; unmapped $j$ imply $\hat f_j = 0$ and provides the same bounds.
The reason for this is simple, since we decrease $\ell$ counters simultaneously, we cannot do this more that $n/\ell$ times.
And since we decrement {\it different} counters, each item's counter is decremented at most $n/\ell$ times.
Variants of this very simple (and very clever) algorithm were independently discovered 
several times~\cite{demaine2002frequency, karp2003simple, golab2003identifying,metwally2006integrated}.\footnote{The reader is referred to~\cite{karp2003simple} for an efficient streaming implementation.}
From this point on, we refer to these collectively as \FI.

Later, Berinde \etal~\cite{BCIS09} proved a tighter bound for \FI.
Consider summing up the errors by $\hat R_k =  \sum_{i=k+1}^{d}  |f_j - \hat f_j |$ and assume without loss of generality that $f_j \geq f_{j+1}$ for all $j$.
Then, it is obvious that counting only the top $k$ items exactly is the best possible strategy if only $k$ counters are allowed.
That is, the optimal solution has a cost of $R_k =  \sum_{i=k+1}^{d}  f_j$. 
Berinde \etal~\cite{BCIS09} showed that if \FI uses $\ell > k$ counters then $|f_j - \hat f_j | \le R_k/(\ell -k)$.
By summing the error over the top $k$ items it is easy to obtain that $\hat R_k < \frac{\ell}{\ell - k}R_k$. 
Setting $\ell  = \lceil k + k/\eps\rceil$ yields the convenient form of $\hat R_k < (1 + \eps)R_k$.
The authors also show that to get this kind of guarantee in the streaming setting $\Omega(k/\eps)$ bits are indeed necessary.
This make \FI optimal up to a $\log$ factor in that regard.

\subsection{Connection to Matrix Sketching}
\label{sec:connection}
There is a tight connection between the matrix sketching problem and the frequent items problem. 
Let $A$ be a matrix that is given to the algorithm as a stream of its rows. 
For now, let us constrain the rows of $A$ to be indicator vectors. 
In other words, we have $a_i \in \{e_1, . . . , e_d\}$, where $e_j$ is the $j$th standard basis vector. 
Note that such a matrix can encode a stream of items (as above). 
If the $i$th element in the stream is $j$, then the $i$th row of the matrix is set to $a_i = e_j$. 
The frequency $f_j$ can be expressed as $f_j = \|Ae_j\|^2$. 
Assume that we construct a matrix $B \in \mathbb{R}^{\ell \times d}$ as follows.
First, we run \FI on the input. Then, for every item $j$ for which $\hat f_j > 0$ we generate one row in $B$ equal to $\hat f_j^{1/2}\cdot e_j$.
The result is a low rank approximation of $A$. 
Note that $\rank(B) = \ell$ and that $\|Be_j\|^2 = \hat f_j$. 
Notice also that $\|A\|_F^2 = n$ and that $A^TA = \diag(f_1,\ldots,f_d)$ and that $B^TB = \diag(\hat f_1,\ldots, \hat f_d)$.
Porting the results we obtained from \FI we get that $\|A^TA - B^TB\|_2 = \max_j | f_j - \hat f_j | \le \|A\|_F^2/(\ell-k)$.
Moreover, since the rows of $A$ (corresponding to different counters) are orthogonal, the best rank $k$ approximation of $A$ would capture exactly the most frequent items.
Therefore, $\|A - A_k\|_F^2 = R_k =  \sum_{i=k+1}^{d}  f_j$. 
If we follow the step above we can also reach the conclusion that $\|A - \pi_B^k(A)\|_F^2 \le \frac{\ell}{\ell - k}\|A - A_k\|_F^2$.
We observe that, for the case of matrices whose rows are basis vectors, \FI actually provides a very efficient low rank approximation result.
In this paper, we argue that an algorithm in the same spirit as \FI obtains the same bounds for general matrices and general test vectors.

\subsection{Main results}
\label{sec:results}
We describe the \FD algorithm, an extension of \FI to general matrices. 
The intuition behind \FD is surprisingly similar to that of \FI: In the same way that \FI periodically deletes $\ell$ different elements, \FD periodically `shrinks' $\ell$ orthogonal vectors by roughly the same amount. This means that during shrinking steps, the squared Frobenius norm of the sketch reduces $\ell$ times faster than its squared projection on any single direction. Since the Frobenius norm of the final sketch is non negative, we are guaranteed that no direction in space is reduced by ``too much". This intuition is made exact below. As a remark, when presented with an item indicator matrix, \FD exactly mimics a variant of \FI.

\begin{theorem}\label{add}
Given any matrix $A \in \R^{n\times d}$, \FD processes the rows of $A$ one by one and produces a sketch matrix $B \in \R^{\ell \times d}$, such that for any \emph{unit vector} $x \in \R^d$ 
\[
0 \leq \|Ax\|^2 - \|Bx\|^2 \leq \|A - A_k\|_F^2/(\ell-k).
\]
This holds also for all $k < \ell$ including $k=0$ where we define $A_0$ as the $n \times d$ all zeros matrix.
\end{theorem}

Other convenient formulations of the above are $A^TA \succeq B^TB \succeq A^TA - I_d \cdot \|A - A_k\|_F^2/(\ell-k)$ or $A^TA - B^TB \succeq 0$ and $\|A^TA - B^TB\|_2 \le \|A - A_k\|_F^2/(\ell-k)$.  Note that setting $\ell = \lceil k + 1/\eps \rceil$ yields error of $\eps  \|A - A_k\|_F^2$ using $O(dk + d/\eps)$ space.  
This gives an additive approximation result which extends and refines that of Liberty~\cite{Lib12}. We also provide a multiplicative error bound from Ghashami and Phillips \cite{GhashamiP14}. 

\begin{theorem}\label{mul}
Let $B \in \R^{\ell \times d}$ be the sketch produced by \FD. For any $k < \ell$ it holds that
\[
\|A - \pi^k_B(A)\|^2_F \leq (1+\frac{k}{\ell - k}) \|A - A_k\|^2_F.
\]
\end{theorem}

Note that by setting  $\ell = \lceil k + k/\eps\rceil$ one gets the standard form $\|A - \pi_B^k(A)\|^2_F \leq (1 + \eps)\|A - A_k\|^2_F$.
\FD achieves this bound while using less space than any other known algorithm, namely $O(\ell d) = O(kd/\eps)$ floating point numbers, and is deterministic.  
Also note that the $\pi_B^k$ operator projects onto a rank $k$ subspace $B_k$, where as other similar bounds\cite{sarlos2006improved,cw09,FSS13,mahoney2009cur,boutsidis2011near} project onto a higher rank subspace $B$, and then considers the best rank $k$ approximation of $\pi_B(A)$.  Our lower bound in Theorem \ref{opt}, from Woodruff \cite{Woo14}, shows our approach is also tight even for this weaker error bound.

\begin{theorem}\label{opt-add}
Assuming a constant word size, any matrix approximation algorithm, 
which guarantees $\|A^TA - B^TB\|_2 \leq \|A-A_k\|^2_F /(\ell-k)$ must use $\Omega(d \ell)$ space; or in more standard terms, guarantees of $\|A^TA - B^TB\|_2 \leq \eps \|A-A_k\|^2_F$ must use $\Omega(d k+d/\eps)$ space.
\end{theorem}

\begin{theorem}\label{opt}
Assuming a constant word size, any randomized matrix approximation streaming algorithm  in the row-wise-updates model, 
which guarantees $\|A - \pi_B^k(A)\|^2_F \leq (1 + \eps)\|A - A_k\|^2_F$ and that succeeds with probability at least $2/3$, must use $\Omega(kd/\eps)$ space.
\end{theorem}

Theorem~\ref{opt} claims that \FD is optimal with respect to the tradeoff between sketch size and resulting accuracy.
On the other hand, in terms of running time, \FD is not known to be optimal. 

\begin{theorem}\label{run}
The running time of \FD on input $A \in \R^{n \times d}$ and parameter $\ell$ is $O(nd\ell)$.
\FD is also ``embarrassingly parallel'' because its resulting sketch matrices constitute \emph{mergeable summaries~\cite{ACHPWY12}}.
\end{theorem}

\subsection{Practical Implications}
\label{sec:implications}
Before we describe the algorithm itself, we point out how it can be used in practice.
As its name suggests, the \FI algorithm is often used to uncover frequent items in an item stream. 
Namely, if one sets $\ell > 1/\eps$, then any item that appears more than $\eps n$ times in the stream must appear in the final sketch. 
Similarly, \FD can be used to discover the space of unit vectors (directions) in space $x$ for which $\|Ax\|^2 \geq \eps\|A\|_F^2$, for example. 
This property makes \FD extremely useful in practice. 
In data mining, it is common to represent data matrices $A$ by their low rank approximations $A_k$.
But, choosing the right $k$ for which this representation is useful is not straight forward. 
If the chosen value of $k$ is too small the representation might be poor. 
If $k$ is too large, precious space and computation cycles are squandered.  
The goal is therefore to pick the minimal $k$ which provides an acceptable approximation.
To do this, as practitioners, we typically compute the top $k' \gg k$ singular vectors and values of $A$ (computing $\svd(A)$ partially).
We then keep only the $k$ singular vectors whose corresponding singular values are larger than some threshold value $t$. 
In other words, we only ``care about" unit vectors $x$ such that $\|Ax\| \geq t$. 
Using \FD we can invert this process. 
We can prescribe in advance the acceptable approximation $t$ and directly find a space containing those vectors $x$ for which $\|Ax\| \geq t$.
Thereby, not only enabling a one-pass or streaming application, but also circumventing the $\svd$ computation altogether.

\section{Frequent Directions}
The algorithm keeps an $\ell \times d$ sketch matrix $B$ that is updated every time a new row from the input matrix $A$ is added. 
The algorithm maintains the invariant that the last row of the sketch $B$ is always all-zero valued.
During the execution of the algorithm, rows from $A$ simply replace the all-zero valued row in $B$. 
Then, the last row is nullified by a two-stage process. 
First, the sketch is rotated (from the left) using its SVD such that its rows are orthogonal and in descending magnitude order. Then, the sketch rows norms are ``shrunk" so that at least one of them is set to zero. 

\begin{algorithm}
\caption{\label{alg:freqDir1} \FD}
\begin{algorithmic}
\STATE \textbf{Input:} $\ell, A \in R^{n \times d}$
\STATE $B \leftarrow 0^{\ell \times d}$
\FOR {$i \in 1, \ldots, n$}
	\STATE $B_\ell \leftarrow a_i$ \hfill \# $i$th row of $A$ replaced (all-zeros) $\ell$th row of $B$
	\STATE $[U,\Sigma, V] \leftarrow \svd(B)$ 
  	\STATE $C \leftarrow \Sigma V^T$ \hfill \# Not computed, only needed for proof notation
  	\STATE $\delta \leftarrow \sigma_{\ell}^2$
  	\STATE $B \leftarrow \sqrt{\Sigma^2 - \delta I_{\ell}} \cdot V^T$ \hfill \# The last row of $B$ is again zero
\ENDFOR
\STATE \textbf{return} $B$
\end{algorithmic}
\end{algorithm}

\subsection{Error Bound}
\label{sec:additive}
This section proves our main results for Algorithm~\ref{alg:freqDir1} which is our simplest and most space efficient algorithm.  
The reader will notice that we occasionally use inequalities instead of equalities at different parts of the proof to obtain three Properties.  
This is not unintentional. 
The reason is that we want the same exact proofs to hold also for Algorithm~\ref{alg:freqDir2} which we describe in Section~\ref{run_forest_run}.
Algorithm~\ref{alg:freqDir2} is conceptually identical to Algorithm~\ref{alg:freqDir1}, it requires twice as much space but is far more efficient.
Moreover, \emph{any algorithm} which produces an approximate matrix $B$ which satisfies the following facts (for any choice of $\Delta$) will achieve the error bounds stated in Lemma \ref{add} and Lemma \ref{mul}.  

In what follows, we denote by $\delta_i$, $B_{[i]}$, $C_{[i]}$ the values of $\delta$, $B$ and $C$ respectively \emph{after} the $i$th row of $A$ was processed.
Let $\Delta = \sum_{i=1}^n \delta_i$, be the total mass we subtract from the stream during the algorithm.
To prove our result we first prove three auxiliary properties.

\begin{property}\label{fact1}
For any vector $x$ we have $\|Ax\|^2 - \|Bx\|^2 \ge 0$.
\end{property}
\begin{proof}
Use the observations that $\langle a_i,x\rangle^2 + \|B_{[i-1]}x\|^2 = \|C_{[i]}x\|^2$.
\[
\|Ax\|^2 - \|Bx\|^2 = \sum_{i=1}^{n}[ \langle a_i,x\rangle^2 + \|B_{[i-1]}x\|^2 - \|B_{[i]}x\|^2] \ge \sum_{i=1}^{n} [\|C_{[i]}x\|^2 - \|B_{[i]}x\|^2]  \ge 0. \qedhere
\]
\end{proof}

\begin{property}\label{fact2}
For any unit vector $x \in \mathbb{R}^d$ we have 
$
\|Ax\|^2 - \|Bx\|^2 \le \Delta.
$
\end{property}
\begin{proof}To see this, first note that $\|C_{[i]} x\|^2 - \|B_{[i]}x\|^2 \le \|C_{[i]}^TC_{[i]} - B_{[i]}^TB_{[i]}\| \le  \delta_i$.
Now, consider the fact that $\|C_{[i]} x \|^2 = \|B_{[i-1]} x \|^2 + \|a_i x\|^2$. 
Substituting for $\|C_{[i]}x\|^2$ above and taking the sum yields
\begin{align}
\nonumber \sum_i \|C_{[i]}x\|^2 - \|B_{[i]}x\|^2 = \sum_i (\|B_{[i-1]} x \|^2 + \|a_i x\|^2) - \|B_{[i]}x\|^2 \\
\nonumber = \|Ax\|^2 + \|B_{[0]}x\|^2 - \|B_{[n]}x\|^2 = \|Ax\|^2 - \|Bx\|^2. 
\end{align}
Combining this with $\sum_i \|C_{[i]}x\|^2 - \|B_{[i]}x\|^2 \le \sum_i \delta_i = \Delta$ yields that $\|Ax\|^2 - \|Bx\|^2 \le \Delta$.
\end{proof}

\begin{property}\label{fact3} 
$\Delta \ell  \le \|A\|_F^2 - \|B\|_F^2$.
\end{property}
\begin{proof}
In the $i$th round of the algorithm $\|C_{[i]}\|^2_F \ge \|B_{[i]}\|^2_F + \ell \delta_i$ and $\|C_{[i]} \|^2_F = \|B_{[i-1]} \|^2_F + \|a_i\|^2$.  
By solving for $\|a_i\|^2$ and summing over $i$ we get
\[
\|A\|^2_F = \sum_{i=1}^n \|a_i\|^2 \le \sum_{i=1}^n \|B_{[i]}\|^2_F - \|B_{[i-1]}\|_F^2 + \ell \delta_i = 
\|B\|^2_F + \ell \Delta. \qedhere
\]
\end{proof}

Equipped with the above observations, and \emph{no additional requirements about the construction of $B$}, we can prove Lemma~\ref{add}.
Namely, that for any $k<\ell$, $\|A^TA - B^TB\|_2 \leq \|A - A_k\|_F^2/(\ell-k)$.
We use Property~\ref{fact2} verbatim and bootstrap Property~\ref{fact3} to prove a tighter bound on $\Delta$.
In the following, $y_i$ correspond to the singular vectors of $A$ ordered with respect to a decreasing corresponding singular values.
\begin{align*}
\Delta \ell	& \le \|A\|_F^2 - \|B\|_F^2 										& \text{via Property~\ref{fact3}} \\ 
		& = \sum_{i=1}^k \|A y_i\|^2 + \sum_{i=k+1}^d \|A y_i\|^2 - \|B\|_F^2 		&  \|A\|_F^2 = \sum_{i=1}^{d}\|Ay_i\|^2\\
		& = \sum_{i=1}^k \|A y_i\|^2 + \|A - A_k \|_F^2 - \|B\|_F^2 				& \\ 
		& \leq \|A - A_k \|_F^2 + \sum_{i=1}^k \left (\|A y_i\|^2 - \|B y_i\|^2 \right) 	& \text{$\sum_{i=1}^k \|B y_i\|^2 < \|B\|_F^2$} \\
		& \leq \|A - A_k \|_F^2 + k\Delta. 								& \text{via Property~\ref{fact2}}			
\end{align*}  
Solving $\Delta \ell \leq \|A - A_k\|_F^2 + k \Delta$ for $\Delta$ to obtain $\Delta \leq \|A - A_k\|_F^2 / (\ell - k)$, which combined with Property~\ref{fact1} and Property~\ref{fact2} proves Theorem \ref{add}:  for any unit vector $x$ we have
\[
0 \leq \|A x\|^2 - \|B x\|^2 \leq \Delta \leq \|A-A_k\|_F^2 / (\ell-k).  
\] 

Now we can show that projecting $A$ onto $B_k$ provides a relative error approximation.  
Here, $y_i$ correspond to the singular vectors of $A$ as above and $v_i$ to the singular vectors of $B$ in a similar fashion.
\begin{align*}
\|A - \pi_{B}^{k}(A)\|_F^2  	&=  \|A\|_F^2 - \|\pi_{B}^{k}(A)\|_F^2 = \|A\|_F^2 - \sum_{i=1}^k \|A v_i\|^2				& \text{Pythagorean theorem} \\ 
					& \leq  \|A\|_F^2 - \sum_{i=1}^k\|B v_i\|^2 										& \text{via Property~\ref{fact1}} \\
					& \leq  \|A\|_F^2 - \sum_{i=1}^k \|B y_i\|^2    									& \text{since $\sum_{i=1}^j \|B v_i\|^2 \geq \sum_{i=1}^j \|B y_i\|^2$} \\ 
					&\leq \|A\|_F^2 - \sum_{i=1}^k (\|A y_i\|^2 - \Delta)  								& \text{via Property~\ref{fact2}} \\ 
					& = \|A\|_F^2 - \|A_k\|_F^2 + k\Delta 											& \\ 
					&\leq \|A - A_k\|_F^2 + \frac{k}{\ell-k} \|A - A_k\|_F^2 								& \text{ by } \Delta \leq \|A - A_k\|_F^2 / (\ell - k) \\
					& = \frac{\ell}{\ell-k} \|A - A_k\|_F^2.  											&
\end{align*}
This concludes the proof of Theorem~\ref{mul}.
It is convenient to set $\ell = \lceil k + k/\eps \rceil$ which results in the standard bound form $\|A - \pi_{B}^{k}(A)\|_F^2 \leq (1+\eps) \|A - A_k\|_F^2$.

\section{Running Time Analysis}\label{run_forest_run}
Each iteration of Algorithm~\ref{alg:freqDir1} is dominated by the computation of the $\svd(B)$.
The standard running time of this operation is $O(d\ell^2)$ \cite{golub2012matrix}. 
Since this loop is executed once per row in $A$ the total running time would na\"ively be $O(nd\ell^2)$.
However, note that the sketch matrix $B$ actually has a very special form. 
The first $\ell-1$ rows of $B$ are always orthogonal to one another. 
This is a result of the sketch having been computed by an $\svd$ in the previous iteration.
Computing the $\svd$ of this matrix is possible in $O(d\ell)$ time using the Gu-Eisenstat procedure \cite{gu93stable}.
This requires using the Fast Multiple Method (FMM) and efficient multiplication of Cauchy matrices by vectors which is, unfortunately, far from being straight forward.
It would have been convenient to use a standard $\svd$ implementation and still avoid the quadratic term in $\ell$ in the running time.
We show below that this is indeed possible at the expense of doubling the space used by the algorithm.
Algorithm~\ref{alg:freqDir2} gives the details.

\begin{algorithm}
\caption{\label{alg:freqDir2} \textsc{Fast}-\FD}
\begin{algorithmic}
\STATE \textbf{Input:} $\ell, A \in R^{n \times d}$
\STATE $B \leftarrow$ all zeros matrix $\in R^{2\ell \times d}$ 
\FOR {$i \in 1,\ldots, n$}
  \STATE  Insert $a_i$ into a zero valued row of $B$
  \IF {$B$ has no zero valued rows}
  	\STATE  $[U,\Sigma, V] \leftarrow \svd(B)$
  	\STATE $C = \Sigma V^T$ \hfill \# Only needed for proof notation
  	\STATE $\delta \leftarrow \sigma_{\ell}^2$ 
  	\STATE $B \leftarrow \sqrt{\max(\Sigma^2 - I_{\ell}\delta , 0)} \cdot V^T$ \hfill \# The last $\ell+1$ rows of $B$ are zero valued.
  \ENDIF 
\ENDFOR
\STATE \textbf{return} $B$
\end{algorithmic}
\end{algorithm}

Note that in Algorithm~\ref{alg:freqDir2} the $\svd$ of $B$ is computed only $n/(\ell +1)$ times because the ``if" statement is only triggered once every $\ell+1$ iterations.
Thereby exhibiting a total running time of $O((n/\ell) d\ell^2) = O(nd\ell)$. 
The reader should revisit the proofs in Section~\ref{sec:additive} and observe that they still hold.
Consider the values of $i$ for which the ``if'' statement is triggered.
It still holds that $0 \preceq C_{[i]}^TC_{[i]} - B_{[i]}^TB_{[i]} \preceq \delta I_d$ and that $\|C_{[i]}\|_F^2 - \|B_{[i]}\|_F^2 \ge \ell\delta$.
For the other values of $i$, the sketch simply aggregates the input rows and there is clearly no incurred error in doing that.
This is sufficient for the same analysis to go through and complete our discussion on the correctness of Algorithm~\ref{alg:freqDir2}.

\subsection{Parallelization and Merging Sketches}
\label{sec:parallelization}
In extremely large datasets, the processing is often distributed among several machines.  
Each machine receives a disjoint input of raw data and is tasked with creating a small space summary.  Then to get a global summary of the entire data, these summaries need to be combined.  
The core problem is illustrated in the case of just two machines, each process a data set $A_1$ and $A_2$, where $A = [A_1; A_2]$, and create two summaries $B_1$ and $B_2$, respectively.   Then the goal is to create a single summary $B$ which approximates $A$ using only $B_1$ and $B_2$.  
If $B$ can achieve the same formal space/error tradeoff as each $B_1$ to $A_1$ in a streaming algorithm, then the summary is called a \emph{mergeable summary}~\cite{ACHPWY12}.  

Here we show that the \FD sketch is indeed mergeable under the following procedure.  Consider $B' = [B_1; B_2]$ which has $2 \ell$ rows; then run \FD (in particular Algorithm \ref{alg:freqDir2}) on $B'$ to create sketch $B$ with $\ell$ rows.  Given that $B_1$ and $B_2$ satisfy Facts \ref{fact1}, \ref{fact2}, and \ref{fact3} with parameters $\Delta_1$ and $\Delta_2$, respectively, we will show that $B$ satisfies the same facts with $\Delta = \Delta_1 + \Delta_2 + \delta$, where $\delta$ is taken from the single shrink operation used in Algorithm \ref{alg:freqDir2}.  This implies $B$ automatically inherits the bounds in Theorem \ref{add} and Theorem \ref{mul} as well.   

First note that $B'$ satisfies all facts with $\Delta' = \Delta_1 + \Delta_2$, by additivity of squared  spectral norm along any direction $x$ (e.g. $\|B_1 x\|^2 + \|B_2 x\|^2 = \|B' x\|^2$)
 and squared Frobenious norms (e.g. $\|B_1\|_F^2 + \|B_2\|_F^2 = \|B'\|_F^2$), but has space twice as large as desired.  	
Property~\ref{fact1} is straight forward for $B$ since $B$ only shrinks in all directions in relation to $B'$.  
For Property~\ref{fact2} follows by considering any unit vector $x$ and expanding 
\[
\|Bx\|^2 \geq \|B'x\|^2 - \delta \geq \|Ax\|^2 - (\Delta_1 + \Delta_2) -\delta = \|Ax\|^2 - \Delta.
\]
Similarly, Property~\ref{fact3} can be seen as 
\[
\|B\|_F^2 \leq \|B'\|_F^2 - \delta \ell \leq \|A\|_F^2 - (\Delta_1 + \Delta_2) \ell - \delta \ell = \|A\|_F^2 - \Delta \ell.
\]  

This property trivially generalizes to any number of partitions of $A$. It is especially useful when the matrix (or data) is distributed across many machines. 
In this setting, each machine can independently compute a local sketch. These sketches can then be combined in an arbitrary order using \FD.

\subsection{Worst case update time}
The total running time of the Algorithm~\ref{alg:freqDir2} is $O(nd\ell)$ and the \emph{amortized} running time per row update is $O(d\ell)$.
However, the worst case update time is still $\Omega(d\ell^2)$ in those cases where the $\svd$ is computed.
Using the fact that \FD sketches are mergeable, we can actually use a simple trick to guarantee a worst case $O(d\ell)$ update time.
The idea is to double the space usage (once again) and hold two sketches, one in `active' mode and one in svd `maintenance' mode.
For any row in the input, we first add it to the active sketch and then spend $O(d\ell)$ floating point operations in completing the $\svd$ of the sketch in maintenance mode.
After $\ell$ updates, the active sketch runs out of space and must go into maintenance mode.
But, in the same time, a total of $O(d\ell^2)$ floating point operations were invested in the inactive sketch which completed its $\svd$ computation.
At this point, we switch the sketch roles and continue. 
Once the entire matrix was processed, we combine the two sketches using their mergeable property.

\section{Space Lower Bounds}
\label{sec:lowerbounds}
In this section we show that \FD is space optimal with respect to the guaranteed accuracy.  
We present nearly-matching lower bounds for each case.  We show the number of \emph{bits} needed is equivalent to $d$ times the number of rows \FD requires.  
We first prove Theorem \ref{opt-add} showing the covariance error bound in Theorem \ref{add} is nearly tight, regardless of streaming issues.

\begin{theorem}\label{lemlem}
Let $B$ be a $\ell \times d$ matrix approximating a $n \times d$ matrix $A$ such that $\|A^TA - B^TB\|_2 \leq \|A-A_k\|_F^2/(\ell-k)$. 
For any algorithm with input as an $n \times d$ matrix $A$, the space complexity of representing $B$ is $\Omega (d\ell)$ bits of space.
\end{theorem}
\begin{proof}
For intuition, consider the set of matrices $\mathcal{A}$ such that for all $A \in \mathcal{A}$ we have $A^TA$ is an $\ell/4$ dimensional projection matrix. For such matrices $\|A^TA\| = 1$ and $\|A-A_k\|_F^2/(\ell-k) = 1/4$.
The condition $\|A^TA - B^TB\|_2 \leq 1/4$ means that $B^TB$ gets ``close to'' $A^TA$ which should intuitively require roughly $\Theta(d\ell)$ bits (which are also sufficient to represent $A^TA$).

To make this argument hold mathematically, consider a set of matrices $\mathcal{A}$ such that $\|A_i^TA_i - A_j^TA_j\| > 1/2$ for all $A_i,A_j \in \mathcal{A}$.
Consider also that the sketching algorithm computes $B$ which corresponds to $A_i \in \mathcal{A}$. 
Since  $\|A^TA - B^TB\|_2 \leq 1/4$ and $\|A_i^TA_i - A_j^TA_j\| > 1/2$ there could be only one index such that $\|A_i^TA_i - B^TB\| \leq 1/4$ by the triangle inequality.
Therefore, since $B$ indexes uniquely into $\mathcal{A}$, it must be that encoding $B$ requires at least $\log(|\mathcal{A}|)$ bits.
To complete the proof we point out that there exists a set $\mathcal{A}$ for which $\|A_i^TA_i - A_j^TA_j\| > 1/2$ for all $A_i,A_j \in \mathcal{A}$ and  $\log(|\mathcal{A}|) = \Omega(d\ell)$. This is proven by Kapralov and Talwar \cite{KapralovT13} using a result of Absil, Edelman, and Koev \cite{AbsEdeKoe2006}.
In \cite{KapralovT13} Corollary $5.1$, setting $\delta=1/2$ and $k=\ell/4$ (assuming $d \ge \ell)$ yields that $|\mathcal{A}| = 2^{\Omega(\ell d)}$ and completes the proof.
\end{proof}
The consequence of Theorem~\ref{lemlem} is that the space complexity of \FD is optimal regardless of streaming issues.
In other words, any algorithm satisfying $\|A^TA - B^TB\|_2 \leq \|A-A_k\|_F^2/(\ell-k)$ must use space $\Omega(\ell d)$ since this is the information lower bound for representing $B$; or any algorithm satisfying $\|A^TA - B^TB\|_2 \leq \eps \|A - A_k\|_F^2$ must have $\ell = \Omega(k + 1/\eps)$ and hence use $\Omega(d \ell) = \Omega(dk + d/\eps)$ space.

Next we will prove Theorem \ref{opt}, showing the subspace bound $\|A - \pi_B^k(A)\|_F^2 \leq (1+\eps) \|A- A_k\|_F^2$ preserved by Theorem \ref{mul} with $\ell = k+ k/\eps$ rows, and hence $O(kd/\eps)$ space, is also nearly tight in the row-update streaming setting.  
In fact it shows that finding $B$ such that $\|A - A B^\dagger B\|_F = \|A - \pi_B(A)\|_F \leq (1+\eps) \|A - A_k\|_F$ requires $\Omega(dk/\eps)$ space in a row-update streaming setting.  
This requires careful invocation of more machinery from communication complexity common in streaming lower bounds, and hence we start with a simple and weaker lower bound, and then some intuition before providing the full construction.

\subsection{A Simple Lower Bound}
We start with a simple intuitive lemma showing an $\Omega(kd)$ 
lower bound, which we will refer to. We then prove our main $\Omega(kd/\epsilon)$ lower bound. 
\begin{lemma}\label{lem:simple}
Any streaming algorithm which, for every input $A$, with constant probability (over its internal randomness) succeeds in   
outputting a matrix $R$ for which $\|A- A R^{\dagger}R\|_F \leq (1+\eps)\|A-A_k\|_F$ must
use $\Omega(kd)$ bits of space. 
\end{lemma}
\begin{proof}
Let $\mathcal{S}$ be the set of $k$-dimensional subspaces over the vector space $GF(2)^d$, where $GF(2)$ denotes the finite
field of $2$ elements with the usual modulo $2$ arithmetic. The cardinality of $\mathcal{S}$
is known \cite{ln86} to be 
$$\frac{(2^d-1)(2^d-1) \cdots (2^{d-k+1}-1)}{(2^k-1)(2^{k-1}-1) \cdots 1} \geq 2^{dk/2 - k^2} \geq 2^{dk/6},$$
where the inequalities assume that $k \leq d/3$. 

Now let $A^1$ and $A^2$ be two $k \times d$ 
matrices with entries in $\{0,1\}$ whose rows span two different $k$-dimensional subspaces of $GF(2)^d$. 
We first claim that the rows also span two different $k$-dimensional subspaces of $\mathbb{R}^d$. 
Indeed, consider a vector $v \in GF(2)^d$
which is in the span of the rows of $A^1$ but not in the span of the rows of $A^2$. If $A^1$ and $A^2$ had the same row
span over $\mathbb{R}^d$, then $v = \sum_i w_i A^2_i$, where the $w_i \in \mathbb{R}$ and $A^2_i$ denotes the $i$-th row of $A^2$. 
Since $v$ has integer coordinates and the $A^2_i$ have integer coordinates, we can assume the $w_i$ are rational, since the
irrational parts must cancel. By scaling by the least common multiple of the denominators of the $w_i$, we obtain that
\begin{eqnarray}\label{eqn:relation} 
\alpha \cdot v = \sum_i \beta_i A^2_i,
\end{eqnarray}
where $\alpha, \beta_1, \ldots, \beta_k$ are integers. We can assume that the greatest
common divisor (gcd) of $\alpha, \beta_1, \ldots, \beta_k$ is $1$, 
otherwise the same conclusion holds after we divide $\alpha, \beta_1, \ldots, \beta_k$
by the gcd. Note that (\ref{eqn:relation}) implies that $\alpha v = \sum_i \beta_i A^2_i \bmod 2$, i.e., when we take
each of the coordinates modulo $2$. Since the $\beta_i$ cannot all be divisible by $2$ (since $\alpha$ would then be odd
and so by the gcd condition 
the left hand side would contain a vector with at least one odd coordinate, contradicting that the right hand
side is a vector with even coordinates), and the rows of $A^2$ form a basis
over $GF(2^d)$, the right hand side must be non-zero, which implies that $\alpha = 1 \bmod 2$. This implies that $v$ is
in the span of the rows of $A^2$ over $GF(2^d)$, a contradiction.  

It follows that there are at least $2^{dk/6}$ distinct $k$-dimensional subspaces of $\mathbb{R}^d$ spanned by the rows of the set
of binary
$k \times d$ matrices $A$. For each such $A$, $\|A-A_k\|_F = 0$ and so the row span of $R$ must agree with the row span of $A$
if the streaming algorithm succeeds. It
follows that the output of the streaming algorithm can be used to encode $\log_2 2^{dk/6} = \Omega(dk)$ bits of information. Indeed,
if $A$ is chosen at random from this set of at least $2^{dk/6}$ binary matrices, and $Z$ is a bit indicating if the streaming
algorithm succeeds, then 
$$|R| \geq H(R) \geq I(R ; A | Z) \geq (2/3) I(R ; A \mid Z = 1) \geq (2/3)(dk/6) = \Omega(dk),$$
where $|R|$ denotes the expected length of the encoding of $R$, $H$ is the entropy function, and $I$ is the mutual information. 
For background on information theory, see \cite{ct91}. This completes the proof. 
\end{proof}

\subsection{Intuition for Main Lower Bound}
The only other lower bounds for streaming algorithms for low rank approximation that we know of are due to 
Clarkson and Woodruff \cite{cw09}. As in their work, we use the {\sf Index} problem in communication complexity to establish
our bounds, which is a communication game between two players Alice and Bob, holding a string $x \in \{0,1\}^r$ and an 
index $i \in [r] \eqdef \{1, 2, \ldots, r\}$, respectively. In this game Alice sends a single message to Bob who should output
$x_i$ with constant probability. It is known (see, e.g., \cite{knr99}) that this problem requires Alice's message to be
$\Omega(r)$ bits long. If {\sf Alg} is a streaming algorithm for low rank approximation, and Alice can create a matrix
$A_x$ while Bob can create a matrix $B_i$ (depending on their respective inputs $x$ and $i$), 
then if from the output of {\sf Alg} on the concatenated matrix $[A_x ; B_i]$ Bob can output $x_i$ with constant probability,
then the memory required of {\sf Alg} is $\Omega(r)$ bits, since Alice's message is the state of {\sf Alg}
after running it on $A_x$. 

The main technical challenges are thus in showing how to choose $A_x$ and $B_i$, as well as showing how the output of {\sf Alg}
on $[A_x ; B_i]$ can be used to solve {\sf Index}. This is where our work departs significantly from that of Clarkson and
Woodruff \cite{cw09}. Indeed, a major challenge is that in Theorem \ref{opt}, we only require the output to be the matrix $R$, whereas
in Clarkson and Woodruff's work from the output one can reconstruct $AR^{\dagger} R$. This causes technical complications,
since there is much less information in the output of the algorithm to use to solve the communication game. 

The intuition behind the proof of Theorem \ref{opt} is that given a $2 \times d$ matrix $A = [1, x; 1, 0^d]$, 
where $x$ is a random unit vector, then if $P = R^{\dagger}R$ is a sufficiently good projection matrix for the low rank approximation
problem on $A$, 
then the second row of $AP$ actually reveals a lot of information about $x$. This may be counterintuitive at first, since one may
think that $[1, 0^d; 1, 0^d]$ is a perfectly good low rank approximation. However, it turns out that $[1, x/2; 1, x/2]$ is a much
better low rank approximation in Frobenius norm, and even this is not optimal. Therefore, Bob, who has $[1, 0^d]$ {\it together}
with the output $P$, can compute the second row of $AP$, which necessarily reveals a lot of information about $x$ (e.g., if
$AP \approx [1, x/2 ; 1, x/2]$, its second row would reveal a lot of information about $x$), and therefore one could hope to embed
an instance of the {\sf Index} problem into $x$.  Most of the technical work 
is about reducing the general problem to this $2 \times d$ primitive problem.

\subsection{Proof of Main Lower Bound for Preserving Subspaces}
Now let $c > 0$ be a small constant to be determined. 
We consider the following two player problem between Alice and Bob: Alice has a $c k/\eps \times d$ matrix 
$A$ which can be written as a block matrix
$[I, R]$, where $I$ is the $c k/\eps \times c k/\eps$ identity matrix, and $R$ is a $c k/\eps \times (d-c k/\eps)$
matrix in which the entries are in $\{-1/(d-c k/\eps)^{1/2}, +1/(d-c k/\eps)^{1/2}\}$. Here $[I, R]$ means we append the
columns of $I$ to the left of the columns of $R$; see Figure \ref{fig:lb}. 
Bob is given a set of $k$ standard unit vectors $e_{i_1}, \ldots, e_{i_k}$, for 
distinct $i_1, \ldots, i_k \in [c k/\eps] = \{1, 2, \ldots, c k/\eps\}$. Here we need $c/\eps > 1$, but we can assume
$\eps$ is less than a sufficiently small constant, as otherwise we would just need to prove an $\Omega(kd)$ lower bound, 
which is established by Lemma \ref{lem:simple}. 

Let $B$ be the matrix $[A; e_{i_1}, \ldots, e_{i_k}]$ obtained by stacking $A$ on top of the vectors $e_{i_1}, \ldots, e_{i_k}$. 
The goal is for Bob to output a rank-$k$ projection matrix $P \in \mathbb{R}^{d \times d}$ for which
$\|B-BP\|_F \leq (1+\eps)\|B-B_k\|_F$. 

Denote this problem by $f$. We will show the randomized $1$-way communication complexity of this problem $R^{1-way}_{1/4}(f)$, 
in which Alice sends a single message to Bob and Bob
fails with probability at most $1/4$, is $\Omega(kd/\eps)$ bits. More precisely, let $\mu$ be the following product
distribution on Alice and Bob's inputs: the entries of $R$ are chosen independently and uniformly at random in 
$\{-1/(d-ck/\eps)^{1/2}, +1/(d-ck/\eps)^{1/2}\}$,
while $\{i_1, \ldots, i_k\}$ is a uniformly random set among all sets of $k$ distinct indices in $[ck/\eps]$. 
We will show that $D^{1-way}_{\mu, 1/4}(f) = \Omega(kd/\eps)$, where $D^{1-way}_{\mu, 1/4}(f)$ denotes the minimum communication cost
over all deterministic $1$-way (from Alice to Bob) protocols which fail with probability at most $1/4$ when the inputs
are distributed according to $\mu$. By Yao's minimax principle (see, e.g., \cite{kn97}), $R^{1-way}_{1/4}(f) \geq D^{1-way}_{\mu,1/4}(f)$. 

We use the following two-player problem {\sf Index} in order to lower bound $D^{1-way}_{\mu,1/4}(f)$. 
In this problem Alice is given a string $x \in \{0,1\}^r$,
while Bob is given an index $i \in [r]$. Alice sends a single message to Bob, who needs to output $x_i$ with probability at
least $2/3$. Again by Yao's minimax principle, we have that $R^{1-way}_{1/3}({\sf Index}) \geq D^{1-way}_{\nu,1/3}({\sf Index})$, 
where $\nu$ is the distribution for which $x$ and $i$ are chosen independently and uniformly at random from
their respective domains. The following is well-known.

\begin{fact}\label{fact:index}\cite{knr99}
$D^{1-way}_{\nu,1/3}({\sf Index}) = \Omega(r).$
\end{fact}

\begin{theorem}
For $c$ a small enough positive constant, and $d \geq k/\eps$, we have $D^{1-way}_{\mu,1/4}(f) = \Omega(dk/\eps)$. 
\end{theorem}
\begin{proof}
We will reduce from the {\sf Index} problem with $r = (ck/\eps)(d-ck/\eps)$. 
Alice, given her string $x$ to {\sf Index}, creates the $ck/\eps \times d$ matrix 
$A = [I, R]$ as follows. The matrix $I$ is the $ck/\eps \times ck/\eps$ identity matrix, while
the matrix $R$ is a $ck/\eps \times (d-ck/\eps)$ matrix with entries in $\{-1/(d-ck/\eps)^{1/2}, +1/(d-ck/\eps)^{1/2}\}$. For an
arbitrary bijection between the coordinates of $x$ and the entries of $R$, Alice sets a given entry in $R$ to $-1/(d-ck/\eps)^{1/2}$
if the corresponding coordinate of $x$ is $0$, otherwise Alice sets the given entry in $R$ to $+1/(d-ck/\eps)^{1/2}$. 
In the {\sf Index} problem, Bob is given an index, which under the bijection between coordinates of $x$ and entries of $R$,
corresponds to being given a row index $i$ and an entry $j$ in the $i$-th row of $R$ that he needs to recover. He
sets $i_{\ell} = i$ for a random $\ell \in [k]$, 
and chooses $k-1$ distinct and random indices $i_j \in [ck/\eps] \setminus \{i_{\ell}\}$, for $j \in [k] \setminus \{\ell\}$. 
Observe that if $(x,i) \sim \nu$, then $(R, i_1, \ldots, i_k) \sim \mu$.
Suppose there is a protocol in which Alice sends a single message to Bob who solves $f$ with probability at least $3/4$ under $\mu$. We show that this can be used to solve
${\sf Index}$ with probability at least $2/3$ under $\nu$. The theorem will follow by Fact \ref{fact:index}.
Consider the matrix $B$ which is the matrix $A$ stacked on top of the rows $e_{i_1}, \ldots, e_{i_k}$, in that order, so that $B$ has $ck/\eps + k$ rows. 

We proceed to lower bound $\|B-BP\|_F^2$ in a certain way, which will allow our reduction to {\sf Index} to be
carried out. We need the following fact:
\begin{fact}((2.4) of \cite{rv10})\label{fact:singular}
Let $A$ be an $m \times n$ matrix with i.i.d. entries which are each $+1/\sqrt{n}$ with probability $1/2$ and 
$-1/\sqrt{n}$ with probability $1/2$, and suppose $m/n < 1$. Then for all $t > 0$,
$$\Pr[\|A\|_2 > 1+ t + \sqrt{m/n}] \leq \alpha e^{-\alpha'n t^{3/2}}.$$
where $\alpha, \alpha' > 0$ are absolute constants. Here $\|A\|_2$ is the operator norm $\sup_x \frac{\|Ax\|}{\|x\|}$ of $A$. 
\end{fact}
We apply Fact \ref{fact:singular} to the matrix $R$, which implies,
$$\Pr[\|R\|_2 > 1 + \sqrt{c} + \sqrt{(ck/\eps)/(d-(ck/\eps))}] \leq \alpha e^{-\alpha'(d-(ck/\eps)) c^{3/4}},$$
and using that $d \geq k/\eps$ and $c > 0$ is a sufficiently small constant, this implies
\begin{eqnarray}\label{eqn:operator}
\Pr[\|R\|_2 > 1 + 3\sqrt{c}] \leq e^{-\beta d},
\end{eqnarray}
where $\beta > 0$ is an absolute constant (depending on $c$).  
Note that for $c > 0$ sufficiently small, $(1+3\sqrt{c})^2 \leq 1+7\sqrt{c}$. Let $\mathcal{E}$ be the event
that $\|R\|_2^2 \leq 1+ 7\sqrt{c}$, which we condition on.

We partition the rows of $B$ into $B_1$ and $B_2$, where $B_1$ contains those
rows whose projection onto the first $ck/\eps$ coordinates equals $e_i$ for some 
$i \notin \{i_1, \ldots, i_k\}$. Note that $B_1$ is $(ck/\eps -k) \times d$ and $B_2$ is $2k \times d$. Here, 
$B_2$ is $2k \times d$ since it includes the rows in $A$ indexed by $i_1, \ldots, i_k$, together with the
rows $e_{i_1}, \ldots, e_{i_k}$. 
Let us also partition the rows of $R$ into $R_T$ and $R_S$, so that the union of the rows in $R_T$ 
and in $R_S$ is equal to $R$, 
where the rows of $R_T$ are the rows of $R$ in $B_1$,
and the rows of $R_S$ are the non-zero rows of $R$ in $B_2$ (note that $k$ of the rows are non-zero and $k$ are zero in $B_2$
restricted to the columns in $R$). 

\begin{figure}
\begin{center}\includegraphics[width=138mm]{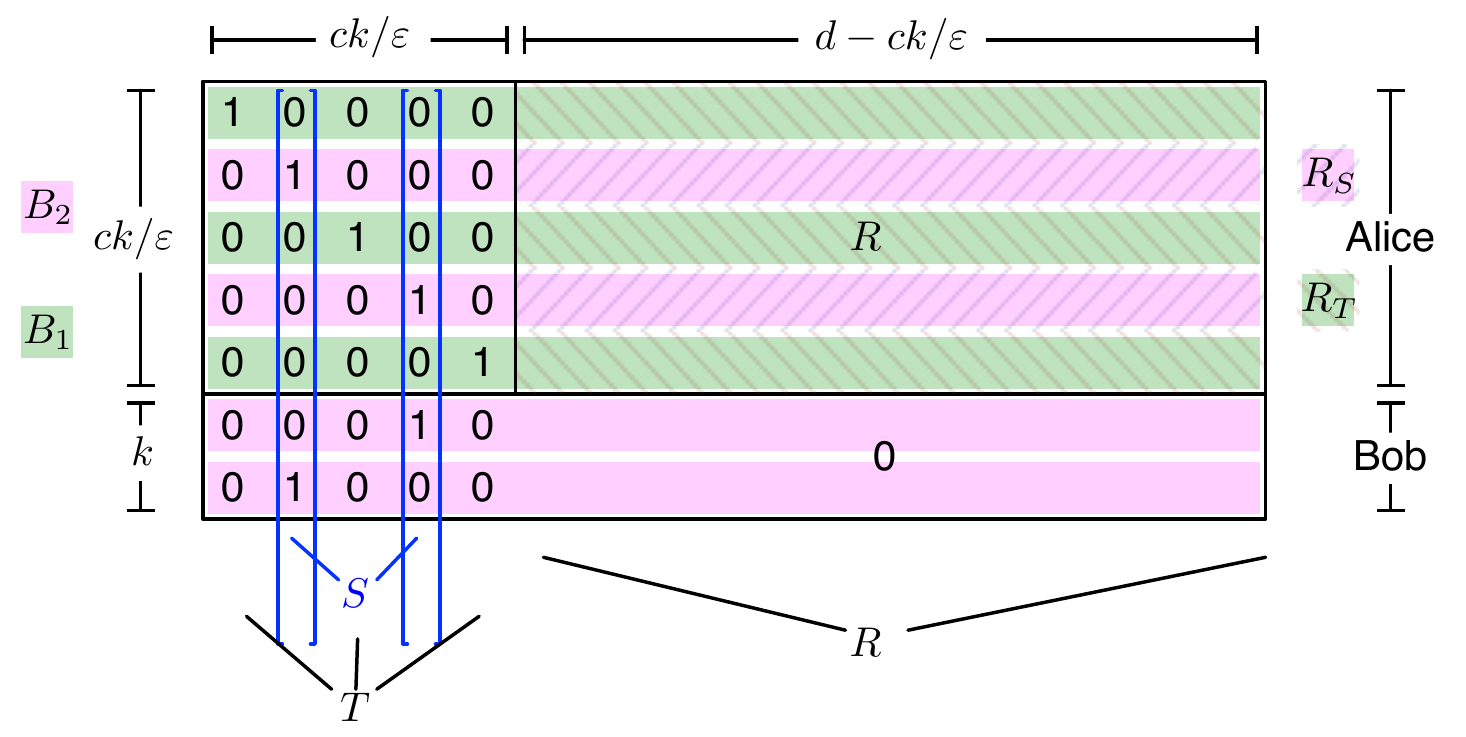}\end{center}
\caption{\label{fig:lb} Illustration of matrix $B$ used for lower bound construction.}
\end{figure}

%
\begin{lemma}\label{lem:technical}
For any unit vector $u$, write $u = u_R + u_S + u_T$, where $S = \{i_1, \ldots, i_k\}, T = [ck/\eps] \setminus S$,
and $R = [d] \setminus [ck/\eps]$, and where $u_A$ for a set $A$ is $0$ on indices $j \notin A$. Then, conditioned
on $\mathcal{E}$ occurring, 
$\|Bu\|^2 \leq (1 + 7\sqrt{c}) (2-\|u_T\|^2 - \|u_R\|^2 + 2\|u_S + u_T\| \|u_R\|).$
\end{lemma}
\begin{proof}

Let $C$ be the matrix consisting of the top $ck/\eps$ rows of $B$, so that 
$C$ has the form $[I, R]$, where $I$ is a $ck/\eps \times ck/\eps$ identity matrix. 
By construction of $B$, $\|Bu\|^2 = \|u_S\|^2 + \|Cu\|^2$. Now, 
$Cu = u_S + u_T + Ru_R$, and so 
\begin{eqnarray*}
\|Cu\|_2^2 & = & \|u_S + u_T\|^2 + \|Ru_R\|^2 + 2(u_s + u_T)^T Ru_R \\
& \leq & \|u_S + u_T\|^2 + (1+7\sqrt{c}) \|u_R\|^2 + 2\|u_S+u_T\| \|Ru_R\| \\
& \leq & (1+7\sqrt{c})(\|u_S\|^2 + \|u_T\|^2 + \|u_R\|^2) + (1+3\sqrt{c}) 2\|u_S + u_T\| \|u_R\| \\
& \leq & (1+7\sqrt{c})(1 + 2\|u_S + u_T\| \|u_R\|),
\end{eqnarray*}
and so 
\begin{eqnarray*}
\|Bu\|^2 & \leq & (1+7\sqrt{c})(1+\|u_S\|^2 + 2\|u_S+u_T\| \|u_R\|)\\
& = & (1+7\sqrt{c})(2-\|u_R\|^2 - \|u_T\|^2 + 2\|u_S+U_T\| \|u_R\|).\qedhere
\end{eqnarray*}
\end{proof}

We will also make use of the following simple but tedious fact.
\begin{fact}\label{fact:maximizer}
For $x \in [0,1]$, the function $f(x) = 2x\sqrt{1-x^2} - x^2$ is maximized when
$x = \sqrt{1/2 - \sqrt{5}/10}$. We define $\zeta$ to be the value of $f(x)$ at its maximum, 
where $\zeta = 2/\sqrt{5} + \sqrt{5}/10 - 1/2 \approx .618$. 
\end{fact}
\begin{proof}
Setting $y^2 = 1-x^2$, we can equivalently maximize $f(y) = -1 + 2y\sqrt{1-y^2} + y^2$, or equivalently
$g(y) = 2y\sqrt{1-y^2} + y^2$. Differentiating this expression and
equating to $0$, we have
$$2\sqrt{1-y^2} - \frac{2y^2}{\sqrt{1-y^2}} + 2y = 0.$$
Multiplying both sides by $\sqrt{1-y^2}$ one obtains the equation
$4y^2 - 2 = 2y \sqrt{1-y^2}$, and squaring both sides, after some algebra one obtains
$5y^4 - 5y^2 + 1 = 0$. Using the quadratic formula, we get that the maximizer satisfies $y^2 = 1/2 + \sqrt{5}/10$, or 
$x^2 = 1/2 - \sqrt{5}/10$. 
\end{proof}

\begin{corollary}\label{fact:operator}
Conditioned on $\mathcal{E}$ occurring, 
$\|B\|_2^2 \leq (1+7\sqrt{c})(2 + \zeta).$
\end{corollary}
\begin{proof}
By Lemma \ref{lem:technical}, for any unit vector $u$, 
$$\|Bu\|^2 \leq (1+7\sqrt{c})(2-\|u_T\|^2 - \|u_R\|^2 + 2\|u_S + u_T\| \|u_R\|).$$
Suppose we replace the vector $u_S + u_T$ with an arbitrary vector supported on coordinates in $S$ with the
same norm as $u_S + u_T$. Then the right hand side of this expression cannot increase, which means it is
maximized when $\|u_T\| = 0$, for which it equals
$(1+7\sqrt{c})(2-\|u_R\|^2 + 2\sqrt{1-\|u_R\|^2}\|u_R\|),$
and setting $\|u_R\|$ to equal the $x$ in Fact \ref{fact:maximizer}, we see that this expression
is at most $(1+7\sqrt{c})(2+\zeta)$. 
\end{proof}

Write the projection matrix $P$ output by the streaming algorithm as $UU^T$, where $U$ is $d \times k$ with orthonormal columns $u^i$
(note that $R^{\dagger} R = P$).
Applying Lemma \ref{lem:technical} and Fact \ref{fact:maximizer} to each of the columns $u^i$, we have:
\begin{eqnarray}
\|BP\|_F^2 & = & \|BU\|_F^2 = \sum_{i=1}^k \|Bu^i\|^2 \nonumber \\
& \leq & (1+7 \sqrt{c}) \sum_{i=1}^k (2-\|u^i_T\|^2 - \|u^i_R\|^2 + 2\|u^i_S + u^i_T\| \|u^i_R\|) \nonumber \\
& = & (1+7\sqrt{c})(2k - \sum_{i=1}^k (\|u^i_T\|^2) + \sum_{i=1}^k (2\|u^i_S + u^i_T\| \|u^i_R\| - \|u^i_R\|^2)) \nonumber \\
& = & (1+7\sqrt{c})(2k - \sum_{i=1}^k (\|u^i_T\|^2) + \sum_{i=1}^k (2 \sqrt{1-\|u^i_R\|^2}\|u^i_R\| - \|u^i_R\|^2)) \nonumber \\
& \leq & (1+7\sqrt{c})(2k - \sum_{i=1}^k (\|u^i_T\|^2) + k \zeta) \nonumber\\
& = & (1+7\sqrt{c}) ((2+\zeta)k - \sum_{i=1}^k \|u^i_T\|^2).\label{eqn:first}
\end{eqnarray}

Using the matrix Pythagorean theorem, we thus have,
\begin{eqnarray}\label{eqn:second}
\|B-BP\|_F^2 & = & \|B\|_F^2 - \|BP\|_F^2 \nonumber\\
& \geq & 2ck/\eps + k - (1+7\sqrt{c}) ((2+\zeta)k - \sum_{i=1}^k \|u^i_T\|^2) \text{ using } \|B\|_F^2 = 2ck/\eps + k \nonumber \\
& \geq & 2ck/\eps + k - (1+7\sqrt{c})(2+\zeta)k + (1+7\sqrt{c})\sum_{i=1}^k \|u^i_T\|^2.
\end{eqnarray}

We now argue that $\|B-BP\|_F^2$ cannot be too large if Alice and Bob succeed in solving $f$. First, we need to upper
bound $\|B-B_k\|_F^2$. To do so, we create a matrix $\tilde{B}_k$ of rank-$k$ and bound $\|B-\tilde{B}_k\|_F^2$. Matrix
$\tilde{B}_k$ will be $0$ on the rows in $B_1$. We can group the rows of $B_2$ into $k$ pairs so that each pair has
the form $e_i + v^i, e_i$, where $i \in [ck/\eps]$ and $v^i$ is a unit vector supported on $[d] \setminus [ck/\eps]$. We let
$Y_i$ be the optimal (in Frobenius norm) rank-$1$ approximation to the matrix $[e_i + v^i; e_i]$. By direct computation
\footnote{For an online SVD calculator, see \url{http://www.bluebit.gr/matrix-calculator/}}
the maximum squared singular value of this matrix is $2 + \zeta$. Our matrix $\tilde{B}_k$ then consists of a single $Y_i$
for each pair in $B_2$. Observe that $\tilde{B}_k$ has rank at most $k$ and
$$\|B-B_k\|_F^2 \leq \|B- \tilde{B}_k\|_F^2 \leq 2ck/\eps + k - (2+\zeta)k,$$
Therefore, if Bob succeeds in solving $f$ on input $B$, then,
\begin{eqnarray}\label{eqn:bound}
\|B-BP\|_F^2 & \leq & (1+\eps)(2ck/\eps +k - (2+\zeta)k) 
\leq 2ck/\eps + k - (2+\zeta)k + 2ck.
\end{eqnarray}
Comparing (\ref{eqn:second}) and (\ref{eqn:bound}), we arrive at, conditioned on $\mathcal{E}$:
\begin{eqnarray}\label{eqn:compare}
\sum_{i=1}^k \|u^i_T\|^2 \leq \frac{1}{1+7\sqrt{c}} \cdot (7\sqrt{c}(2+\zeta)k + 2ck) \leq c_1 k,
\end{eqnarray}
where $c_1 > 0$ is a constant that can be made arbitrarily small by making $ c > 0$ arbitrarily small. 

Since $P$ is a projector, $\|BP\|_F = \|BU\|_F$. Write $U = \hat{U} + \bar{U}$, where
the vectors in $\hat{U}$ are supported on $T$, and the vectors in $\bar{U}$ are supported on $[d] \setminus T$. 
We have,
\begin{eqnarray*}
\|B\hat{U}\|_F^2 \leq \|B\|_2^2 c_1 k \leq (1+7\sqrt{c}) (2+\zeta) c_1 k \leq c_2 k,
\end{eqnarray*}
where the first inequality uses $\|B\hat{U}\|_F \leq \|B\|_2 \|\hat{U}\|_F$ and (\ref{eqn:compare}), 
the second inequality uses that event $\mathcal{E}$ occurs, 
and the third inequality holds for a constant $c_2 > 0$ that can be made arbitrarily small by making the constant
$c > 0$ arbitrarily small. 

Combining with (\ref{eqn:bound}) and using the triangle inequality,
\begin{eqnarray}\label{eqn:mainFirst}
\|B\bar{U}\|_F & \geq & \|BP\|_F - \|B\hat{U}\|_F \hspace{1.1in} \text{ using the triangle inequality}\nonumber \\
& \geq & \|BP\|_F - \sqrt{c_2 k} \hspace{1.2in} \text{ using our bound on }\|B\hat{U}\|_F^2 \nonumber \\
& = & \sqrt{\|B\|_F^2 - \|B-BP\|_F^2} - \sqrt{c_2 k} \hspace{.2in} \text{ by the matrix Pythagorean theorem} \nonumber \\
& \geq & \sqrt{(2+\zeta)k - 2ck} - \sqrt{c_2 k} \hspace{.61in} \text{ by (\ref{eqn:bound})}\nonumber \\
& \geq & \sqrt{(2+\zeta)k - c_3 k},
\end{eqnarray}
where $c_3 > 0$ is a constant that can be made arbitrarily small for $c > 0$ an arbitrarily small constant (note that $c_2 > 0$
also becomes arbitrarily small as $c > 0$ becomes arbitrarily small). Hence, $\|B\bar{U}\|_F^2 \geq (2+\zeta)k - c_3k$,  
and together with Corollary \ref{fact:operator},
that implies $\|\bar{U}\|_F^2 \geq k-c_4k$ for a constant $c_4$ that can be made arbitrarily small by making 
$c > 0$ arbitrarily small. 

Our next goal is to show that $\|B_2\bar{U}\|_F^2$ is almost as large as $\|B\bar{U}\|_F^2$. 
Consider any column $\bar{u}$ of $\bar{U}$, and write it as $\bar{u}_S + \bar{u}_R$. 
Hence,
\begin{eqnarray*}
\|B\bar{u}\|^2 & = & \|R_T\bar{u}_R\|^2 + \|B_2 \bar{u}\|^2 \hspace{1.35in}  \text{ using } B_1\bar{u} = R_T\bar{u}_R\\
& \leq & \|R_T\bar{u}_R\|^2 + \|\bar{u}_S + R_S \bar{u}_R\|^2 + \|\bar{u}_S\|^2 \hspace{.32in} \text{ by definition of the components}\\
& = & \|R\bar{u}_R\|^2 + 2 \|\bar{u}_S\|^2 + 2 \bar{u}_S^T R_S \bar{u}_R \hspace{.65in} \text{ using the Pythagorean theorem}\\
& \leq & 1+7\sqrt{c} + \|\bar{u}_S\|^2 + 2 \|\bar{u}_S\| \|R_S \bar{u}_R\|\\
&& \text{using } \|R\bar{u}_R\|^2 \leq (1+7\sqrt{c})\|\bar{u}_R\|^2 \text{ and } \|\bar{u}_R\|^2 + \|\bar{u}_S\|^2 \leq 1\\
&& \text{(also using Cauchy-Schwarz to bound the other term)}.
\end{eqnarray*}
Suppose $\|R_S \bar{u}_R\| = \tau \|\bar{u}_R\|$ for a value $0 \leq \tau \leq 1+7\sqrt{c}$. Then
$$\|B\bar{u}\|^2 \leq 1 + 7\sqrt{c} + \|\bar{u}_S\|^2 + 2 \tau \|\bar{u}_S\| \sqrt{1-\|\bar{u}_S\|^2}.$$
We thus have, 
\begin{eqnarray}\label{eqn:god}
\|B\bar{u}\|^2 & \leq & 1 + 7\sqrt{c} + (1-\tau)\|\bar{u}_S\|^2 + \tau (\|\bar{u}_S\|^2 + 2 \|\bar{u}_S\| \sqrt{1-\|\bar{u}_S\|^2}) \nonumber \\
& \leq & 1 + 7\sqrt{c} + (1-\tau) + \tau(1 + \zeta) \text{ by Fact \ref{fact:maximizer}} \nonumber \\
& \leq & 2 + \tau \zeta + 7\sqrt{c},
\end{eqnarray}
and hence, letting $\tau_1, \ldots, \tau_k$ denote the corresponding values of $\tau$ for the $k$ columns of $\bar{U}$, we have
\begin{eqnarray}\label{eqn:god2}
\|B\bar{U}\|_F^2 \leq (2 + 7\sqrt{c})k + \zeta \sum_{i=1}^k \tau_i.
\end{eqnarray}
Comparing the square of (\ref{eqn:mainFirst}) with (\ref{eqn:god2}), we have
\begin{eqnarray}\label{eqn:beta}
\sum_{i=1}^k \tau_i & \geq & k - c_5k,
\end{eqnarray}
where $c_5 > 0$ is a constant that can be made arbitrarily small by making $c > 0$ an arbitrarily small constant. 
Now, $\|\bar{U}\|_F^2 \geq k-c_4 k$ as shown above, while since $\|R_s\bar{u}_R\| = \tau_i \|\bar{u}_R\|$ if $\bar{u}_R$ is the $i$-th column
of $\bar{U}$, by (\ref{eqn:beta}) we have 
\begin{eqnarray}\label{eqn:dead}
\|R_S \bar{U}_R \|_F^2 \geq (1-c_6)k
\end{eqnarray}
for a constant $c_6$ that can be made arbitrarily small by making $c > 0$ an arbitarily
small constant. 

Now $\|R \bar{U}_R\|_F^2 \leq (1+7\sqrt{c})k$ since event $\mathcal{E}$ occurs, 
and $\|R \bar{U}_R\|_F^2 = \|R_T \bar{U}_R\|_F^2 + \|R_S \bar{U}_R\|_F^2$ since the rows of $R$ are the
concatenation of rows of $R_S$ and $R_T$, so combining with (\ref{eqn:dead}), we arrive at 
\begin{eqnarray}\label{eqn:statement}
\|R_T \bar{U}_R \|_F^2 &\leq &c_7 k,
\end{eqnarray}
for a constant $c_7 > 0$ that can be made arbitrarily small by making $c > 0$ arbitrarily small. 

Combining the square of (\ref{eqn:mainFirst}) with (\ref{eqn:statement}), we thus have 
\begin{eqnarray}\label{eqn:toB2}
\|B_2\bar{U}\|_F^2 & = & \|B \bar{U}\|_F^2 - \|B_1 \bar{U}\|_F^2 
= \|B \bar{U}\|_F^2 - \|R_T \bar{U}_R\|_F^2 
\geq (2+\zeta)k - c_3 k - c_7 k \notag \\
& \geq & (2+\zeta)k - c_8 k,
\end{eqnarray}
where the constant $c_8 > 0$ can be made arbitrarily small by making $c > 0$ arbitrarily small.

By the triangle inequality, 
\begin{eqnarray}\label{eqn:above}
\|B_2U\|_F \geq \|B_2 \bar{U}\|_F - \|B_2 \hat{U}\|_F \geq ((2+\zeta)k - c_8k)^{1/2} - (c_2k)^{1/2}.
\end{eqnarray}
Hence,
\begin{eqnarray}\label{eqn:local}
\|B_2 -B_2 P\|_F & = & \sqrt{\|B_2\|_F^2 - \|B_2 U\|_F^2} \hspace{.6in} \text{ by matrix Pythagorean, }\|B_2U\|_F = \|B_2P\|_F \nonumber \\
& \leq & \sqrt{\|B_2\|_F^2 - (\|B_2 \bar{U}\|_F - \|B_2 \hat{U}\|_F)^2} \hspace{1.1in} \text{ by triangle inequality} \nonumber\\
& \leq & \sqrt{3k - (((2+\zeta)k - c_8k)^{1/2} - (c_2k)^{1/2})^2} \hspace{.5in} \text{ using (\ref{eqn:above}), } \|B_2\|_F^2 = 3k,\end{eqnarray}
or equivalently, $\|B_2 - B_2 P\|_F^2 \leq 3k - ((2+\zeta)k - c_8k) - (c_2k) + 2k(((2+\zeta)-c_8)c_2)^{1/2}
\leq (1-\zeta)k + c_8k + 2k(((2+\zeta)-c_8)c_2)^{1/2} \leq (1-\zeta)k + c_9k$ for a constant $c_9 > 0$ that can be made
arbitrarily small by making the constant $c > 0$ small enough.  
This intuitively says that $P$ provides
a good low rank approximation for the matrix $B_2$. Notice that by (\ref{eqn:local}),
\begin{eqnarray}\label{eqn:project}
\|B_2P\|_F^2 & = & \|B_2\|_F^2 - \|B_2 - B_2P\|_F^2 
\geq 3k - (1-\zeta) k - c_9k 
\geq (2+\zeta)k - c_9k.
\end{eqnarray}
Now $B_2$ is a $2k \times d$ matrix and we can partition its rows into $k$ pairs of rows
of the form $Z_{\ell} = (e_{i_\ell} + R_{i_{\ell}}, e_{i_{\ell}})$, for $\ell = 1, \ldots, k$. Here we abuse notation and think of $R_{i_{\ell}}$
as a $d$-dimensional vector, its first $ck/\eps$ coordinates set to $0$. Each such pair of rows is a rank-$2$ matrix, which we abuse
notation and call $Z_{\ell}^T$. By direct computation
\footnote{We again used the calculator at \url{http://www.bluebit.gr/matrix-calculator/}} 
$Z_{\ell}^T$ has squared maximum singular value $2+\zeta$. We would like to argue that the projection of $P$
onto the row span of most $Z_{\ell}$ has length very close to $1$. To this end, for each $Z_{\ell}$ consider the orthonormal basis $V_{\ell}^T$
of right singular vectors for its row space (which is span($e_{i_{\ell}}, R_{i_{\ell}}$)). We let $v_{\ell, 1}^T, v_{\ell, 2}^T$ be these two right
singular vectors with corresponding singular values $\sigma_1$ and $\sigma_2$ (which will be the same for all $\ell$, see below). 
We are interested in the quantity $\Delta = \sum_{\ell = 1}^k \|V_{\ell}^T P\|_F^2$ which
intuitively measures how much of $P$ gets projected onto the row spaces of the $Z_{\ell}^T$. 

\begin{lemma}\label{lem:delta}
Conditioned on event $\mathcal{E}$, $\Delta \in [k -c_{10} k, k+c_{10}k],$ where $c_{10} > 0$ is a constant that can be made
arbitrarily small by making $c > 0$ arbitrarily small. 
\end{lemma}
\begin{proof}
For any unit vector $u$, consider $\sum_{\ell = 1}^k \|V_{\ell}^T u\|^2$. This is equal to $\|u_S\|^2 + \|R_S u_R\|^2$. Conditioned on
$\mathcal{E}$, $\|R_S u_R\|^2 \leq (1+7\sqrt{c})\|u_R\|^2$. Hence, $\sum_{\ell = 1}^k \|V_{\ell}^T u\|^2 \leq 1+7\sqrt{c}$, and consequently,
$\Delta \leq k(1+7\sqrt{c})$. 

On the other hand, $\|B_2P\|_F^2 = \sum_{\ell = 1}^k \|Z_{\ell}^T P\|_F^2$. Since $\|Z_{\ell}^T\|_2^2 \leq 2+\zeta$, it follows by (\ref{eqn:project})
that $\Delta \geq k-(c_9/(2+\zeta))k$, as otherwise $\Delta$ would be too small in order for (\ref{eqn:project}) to hold. 

The lemma now follows since $\sqrt{c}$ and $c_9$ can be made arbitrarily small by making the constant $c > 0$ small enough. 
\end{proof}

We have the following corollary.
\begin{corollary}\label{cor:blah}
Conditioned on event $\mathcal{E}$, for a $1-\sqrt{c_9 + 2c_{10}}$ fraction of $\ell \in [k]$,
$\|V_{\ell}^TP\|_F^2 \leq 1 + c_{11}$, 
and for a $99/100$ fraction of $\ell \in [k]$, we have
$\|V_{\ell}^T P\|_F^2 \geq 1 - c_{11},$
where $c_{11} > 0$ is a constant that can be made arbitrarily small by making the constant $c > 0$ arbitrarily small. 
\end{corollary}
\begin{proof}
For the first part of the corollary, observe that 
$$\|Z_{\ell}^T P\|_F^2 = \sigma_1^2 \|v_{\ell, 1}^T P\|^2 + \sigma_2^2 \|v_{\ell, 2}^TP\|^2,$$
where $v_{\ell, 1}^T$ and $v_{\ell, 2}^T$ are right singular vectors of $V_{\ell}^T$, and $\sigma_1$, $\sigma_2$
are its singular values, with $\sigma_1^2 = 2+\zeta$ and $\sigma_2^2 = 1-\zeta$. Since
$\Delta \leq k + c_{10}k$ by Lemma \ref{lem:delta}, we have
$$\sum_{\ell=1}^k \|v_{\ell, 1}^T P\|^2 + \|v_{\ell, 2}^TP\|^2 \leq k + c_{10}k.$$
If $\sum_{\ell = 1}^k \|v_{\ell, 2}^TP\|^2 \geq (c_9 + 2c_{10})k$, then
\begin{eqnarray*}
\|B_2P\|_F^2 & \leq & \sum_{\ell} \|Z_{\ell}^T P\|_F^2\\
& \leq & (2+\zeta)(k + c_{10}k - 2c_{10}k - c_9k) + (1-\zeta) (2c_{10}k + c_9k)\\
& \leq & (2+\zeta)(k - c_9k) - (2+\zeta) c_{10}k + (1-\zeta)(2c_{10}k + c_9k)\\
& \leq & (2+\zeta)k - 2c_9k - \zeta c_9k - 2c_{10}k - \zeta c_{10} k + 2c_{10}k + c_9k - 2\zeta c_{10}k - \zeta c_9k\\
& \leq & (2+\zeta)k - (1+2\zeta)c_9k + -3 \zeta c_{10}k\\
& < & (2+\zeta)k - c_9k 
\end{eqnarray*}
which is a contradiction to (\ref{eqn:project}). 
Hence, $\sum_{\ell = 1}^k \|v_{\ell, 2}^TP\|^2 \leq (c_9 + 2c_{10}) k$. This means by a Markov bound that a $1-\sqrt{c_9 + 2c_{10}}$ fraction
of $\ell$ satisfy $\|v_{\ell, 2}^TP\|^2 \leq \sqrt{c_9 + 2c_{10}}$, which implies that for this fraction 
that $\|V_{\ell}^TP\|_F^2 \leq 1 + \sqrt{c_9 + 2c_{10}}$. 

For the second part of the corollary, suppose at most $99k/100$ different $\ell$ satisfy
$\|V_{\ell}^TP\|_F^2 > 1-200\sqrt{c_9 + 2c_{10}}$. By the previous part of the corollary, at most $\sqrt{c_9 + 2c_{10}}k$
of these $\ell$ can satisfy $\|V_{\ell}^TP\|_F^2 > 1+\sqrt{c_9 + 2c_{10}}$. Hence, since $\|V_{\ell}^TP\|_F^2 \leq 2$, 
\begin{eqnarray*}
\Delta & < & 2\sqrt{c_9 + 2c_{10}}k + (1+\sqrt{c_9 + 2c_{10}})(99/100-\sqrt{c_9 + 2c_{10}})k + (1-200\sqrt{c_9 + 2c_{10}})k/100\\
& \leq & 2\sqrt{c_9 + 2c_{10}}k + 99k/100 + 99k\sqrt{c_9 + 2c_{10}}/100 - k\sqrt{c_9 + 2c_{10}} + k/100 -2\sqrt{c_9 + 2c_{10}}k\\
& \leq & k - \sqrt{c_9 + 2c_{10}}k/100\\
& \leq & k - \sqrt{2c_{10}}k/100\\
& < & k - c_{10} k,
\end{eqnarray*}
where the final inequality follows for $c_{10} > 0$ a sufficiently small constant. 
This is a contradiction to Lemma \ref{lem:delta}. Hence, at least $99k/100$ different $\ell$ satisfy
$\|V_{\ell}^TP\|_F^2 > 1-200\sqrt{c_9 + 2c_{10}}$. Letting $c_{11} = 200\sqrt{c_9 + 2c_{10}}$, we see that $c_{11}$ can be made
an arbitrarily small constant by making the constant $c > 0$ arbitrarily small. This completes the proof. 
\end{proof}

Recall that Bob holds $i = i_{\ell}$ for a random $\ell \in [k]$. It follows (conditioned on $\mathcal{E}$) 
by a union bound that
with probability at least $49/50$, $\|V_{\ell}^TP\|_F^2 \in [1 -c_{11}, 1+c_{11}]$, which we call the event $\mathcal{F}$ and condition on. 
We also condition on event $\mathcal{G}$ that $\|Z_{\ell}^T P \|_F^2 \geq (2+\zeta) - c_{12}$, for a constant $c_{12} > 0$ that can be made
arbitrarily small by making $c > 0$ an arbitrarily small constant.
Combining the first part of Corollary \ref{cor:blah} together with (\ref{eqn:project}), event $\mathcal{G}$ holds with probability
at least $99.5/100$, provided $c > 0$ is a sufficiently small constant. By a union bound it follows that $\mathcal{E}$, $\mathcal{F}$,
and $\mathcal{G}$ occur simultaneously with probability at least $49/51$. 

As $\|Z_{\ell}^T P\|_F^2 = \sigma_1^2 \|v_{\ell, 1}^TP\|^2 + \sigma_2^2 \|v_{\ell, 2}^TP\|^2$, with $\sigma_1^2 = 2+\zeta$ and $\sigma_1^2 = 1-\zeta$, events
$\mathcal{E}, \mathcal{F}$, and $\mathcal{G}$ imply that
$\|v_{\ell,1}^TP\|^2 \geq 1-c_{13}$, where $c_{13} > 0$ is a constant that can be made arbitrarily small by making the constant $c > 0$ arbitrarily
small. Observe that $\|v_{\ell,1}^T P\|^2 = \langle v_{\ell,1}, z \rangle^2$, where $z$ is a unit vector in the direction of the projection
of $v_{\ell,1}$ onto $P$. 

By the Pythagorean theorem, $\|v_{\ell,1}- \langle v_{\ell,1}, z \rangle z\|^2 = 1 -\langle v_{\ell,1}, z \rangle^2$, and so  
\begin{eqnarray}\label{eqn:norm}
\|v_{\ell,1} - \langle v_{\ell, 1}, z \rangle z\|^2 \leq c_{14},
\end{eqnarray}
for a constant $c_{14} > 0$ that can be made arbitrarily small by making $c> 0$ arbitrarily small. 

We thus have $Z_{\ell}^T P = \sigma_1 \langle v_{\ell,1}, z \rangle u_{\ell,1} z^T + \sigma_2 \langle v_{\ell,2}, w \rangle u_{\ell,2} w^T$, 
where $w$ is a unit vector in the direction of the projection of 
of $v_{\ell,2}$ onto $P$, and $u_{\ell,1}, u_{\ell,2}$ are the left singular vectors of $Z_{\ell}^T$. 
Since $ \mathcal{F}$ occurs, we have that $|\langle v_{\ell,2}, w \rangle| \leq c_{11}$, where $c_{11} > 0$ is a 
constant that can be made arbitrarily small by making the constant $c > 0$ arbitrarily small. 
It follows now by (\ref{eqn:norm}) that
\begin{eqnarray}\label{matrixNorm}
\|Z_{\ell}^T P - \sigma_1 u_{\ell,1} v_{\ell,1}^t\|_F^2 \leq c_{15},
\end{eqnarray}
where $c_{15} > 0$ is a constant that can be made arbitrarily small
by making the constant $c > 0$ arbitrarily small. 

By direct calculation\footnote{Using the online calculator in earlier footnotes.}, 
$u_{\ell,1} = -.851e_{i_{\ell}} -.526R_{i_{\ell}}$ and $v_{\ell,1} = -.851e_{i_{\ell}} -.526R_{i_{\ell}}$. 
It follows that $\|Z_{\ell}^T P - (2+\zeta) [.724e_{i_{\ell}} + .448 R_{i_{\ell}}; .448e_{i_{\ell}} + .277R_{i_{\ell}}]\|_F^2 \leq c_{15}$. 
Since $e_{i_{\ell}}$ is the second row of $Z_{\ell}^T$, it follows that
$\|e_{i_{\ell}}^T P - (2+\zeta)(.448e_{i_{\ell}} + .277R_{i_{\ell}})\|^2 \leq c_{15}.$

Observe that Bob has $e_{i_{\ell}}$ and $P$, and can therefore compute $e_{i_{\ell}}^T P$. Moreover, as $c_{15} > 0$ can be made arbitrarily small
by making the constant $c > 0$ arbitrarily small, it follows that a $1-c_{16}$ fraction of the signs of coordinates of $e_{i_{\ell}}^TP$,
restricted to coordinates in $[d] \setminus [ck/\eps]$, must agree with those of $(2+\zeta).277R_{i_{\ell}}$, which in turn agree
with those of $R_{i_{\ell}}$. Here $c_{16} > 0$ is a constant that can be made arbitrarily small by making the constant $c > 0$ arbitrarily
small. Hence, in particular, the sign of the $j$-th coordinate of $R_{i_{\ell}}$, which Bob needs to output, agrees with that of 
the $j$-th coordinate of $e_{i_{\ell}}^TP$ with probability at least $1-c_{16}$. Call this event $\mathcal{H}$.

By a union bound over the occurrence of events $\mathcal{E}, \mathcal{F}$, $\mathcal{G}$, and $\mathcal{H}$, and the streaming
algorithm succeeding (which occurs with probability $3/4$),
it follows that Bob succeeds in solving {\sf Index} with probability at least $49/51 - 1/4 - c_{16} > 2/3$, as required. This completes
the proof. 
\end{proof}

\section{Related Work on Matrix Sketching and Streaming}
\label{sec:related}
As mentioned in the introduction, there are a variety of techniques to sketch a matrix.  There are also several ways to measure error and models to consider for streaming.  

In this paper we focus on the row-update streaming model where each stream elements appends a row to the input matrix $A$.   
A more general model the entry-update model fixes the size of $A$ at $n \times d$, but each element indicates a single matrix entry and adds (or subtracts) to its value.  

\paragraph{Error measures.}
The accuracy of a sketch matrix $B$ can be measured in several ways.  
Most commonly one considers an $n \times d$, rank $k$ matrix $\hat A$ that is derived from $B$ (and sometimes $A$) and measures the \emph{projection error} where \s{proj-err} $= \|A - \hat A\|_F^2/\|A - A_k\|_F^2$.  When $\hat A$ can be derived entirely from $B$, we call this a \emph{construction} result, and clearly requires at least $\Omega(n+d)$ space.  When the space is required to be independent of $n$, then $\hat A$ either implicitly depends on $A$, or requires another pass of the data.  It can then be defined one of two ways: 
$\hat A = \pi_{B}^k(A)$ (as is considered in this paper) takes $B_k$, the best rank-$k$ approximation to $B$, and then projects $A$ onto $B_k$;
$\hat A = \Pi_{B}^k(A)$ projects $A$ onto $B$, and then takes the best rank-$k$ approximation of the result.  Note that $\pi_{B}^k(A)$ is better than $\Pi_{B}^k(A)$, since it knows the rank-$k$ subspace to project onto without re-examining $A$.  

We also consider \emph{covariance error} where \s{covar-err} $= \|A^TA-B^TB\|_2/\|A\|_F^2$ in this paper.  One can also bound $\|A^TA-B^TB\|_2/\|A-A_k\|_F^2$, but this has an extra parameter $k$, and is less clean.  This measure captures the norm of $A$ along all directions (where as non-construction, projection error only indicates how accurate the choice of subspace is), but still does not require $\Omega(n)$ space.  

\paragraph{Sketch paradigms}
Given these models, there are several types of matrix sketches.  We describe them here with a bit more specificity than in the Introduction, with particular attention to those that can operate in the row-update model we focus on.  Specific exemplars are described which are used in out empirical study to follow.    
The first approach is to \emph{sparsify} the matrix~\cite{arora2006fast,achlioptas2001fast,drineas2011note}, by retaining a small number of non-zero.  These algorithms typically assume to know the $n \times d$ dimensions of $A$, and are thus not directly applicable in out model.  

\noindent  \s{Random-projection}: The second approach randomly combines rows of the matrix~\cite{papadimitriou1998latent,vempala2004random,sarlos2006improved,liberty2007randomized}.
For an efficient variant~\cite{achlioptas2001database} we will consider where the sketch $B$ is equivalent to $R A$ where $R$ is an $\ell \times n$ matrix such that each element $R_{i,j} \in \{-1/\sqrt{\ell},1/\sqrt{\ell}\}$ uniformly.  This is easily computed in a streaming fashion, while requiring at most $O(\ell d)$ space and $O(\ell d)$ operation per row updated.  
Sparser constructions of random projection matrices are known to exist~\cite{kane2012sparser,dasgupta2010sparse}. These, however, were not implemented since the running time of applying random projection matrices is not the focus of this experiment.

\noindent  \s{Hashing}: A variant of this approach~\cite{clarkson2013low} uses an extra sign-hash function to replicate the count-sketch~\cite{charikar2002finding} with matrix rows (analogously to how \FD does with the MG sketch).  
Specifically, the sketch $B$ is initialized as the $\ell \times d$ all zeros matrix, then each row $a_i$ of $A$ is added to row $h(i)$ as $B_{h(i)} \leftarrow B_{h(i)}+ s(i)a_i$, where $h : [n] \rightarrow [\ell]$ and $s : [n] \leftarrow \{-1,1\}$ are perfect hash functions. There is no harm in assuming such functions exist since complete randomness is na\"{\i}vely possible without dominating either space or running time. This method is often used in practice by the machine learning community and is referred to as ``feature hashing"~\cite{weinberger2009feature}.  
Surprising new analysis of this method~\cite{clarkson2013low} shows this approach is optimal for construction bounds and takes $O(nnz(A))$ processing time for any matrix $A$.  

\noindent \s{Sampling}: 
The third sketching approach is to find a small subset of matrix rows (and/or columns) that approximate the entire matrix. This problem is known as the `Column Subset Selection Problem' and has been thoroughly investigated~\cite{frieze2004fast,drineas2003pass,boutsidis2009improved,deshpande2006adaptive,drineas2011faster,boutsidis2011near}. 
Recent results offer algorithms with almost matching lower bounds~\cite{cw09,boutsidis2011near,deshpande2006adaptive}. 
A simple streaming solution to the `Column Subset Selection Problem' is obtained by sampling rows from the input matrix with probability proportional to their squared $\ell_2$ norm. 
Specifically, each row $B_j$ takes the value $A_i/\sqrt{\ell p_i}$ iid with probability $p_i = \|A_i\|^2/\|A\|_F^2$.
The space it requires is $O(\ell d)$ in the worst case but it can be much lower if the chosen rows are sparse. Since the value of $\|A\|_F$ is not a priori known, the streaming algorithm is implemented by $\ell$ independent reservoir samplers, each sampling a single row according to the distribution. The update running time is therefore $O(d)$ per row in $A$. 
Despite this algorithm's apparent simplicity, providing tight bounds for its error performance required over a decade of research~\cite{frieze2004fast,ahlswede2002strong,drineas2003pass,rudelson2007sampling,vershynin2009note,oliveira2010sums,drineas2011faster}.  
Advanced such algorithms utilize the leverage scores of the rows~\cite{drineas2008relative} and not their squared $\ell_2$ norms. The discussion on matrix leverage scores goes beyond the scope of this paper, see~\cite{drineas2011faster} for more information and references.

There are a couple of other sketching techniques that try to maintain some form of truncated SVD as the data is arriving.  Incremental SVD and its variants~\cite{golub2012matrix, hall1998incremental, levey2000sequential, brand2002incremental,ross2008incremental}, are quite similar to \FD, but do not perform the shrinkage step; with each new row they recompute the truncated SVD, simply deleting the smallest singular value/vector.  No worst case error bounds can be shown for this approach, and a recent paper~\cite{GDP14} demonstrates how it may fail dramatically; this paper~\cite{GDP14} (which appeared after the initial conference versions of this paper) also shows how to approach or surpass the performance of Incremental SVD approaches using variants of the \FD approach described herein.  
Another approach by Feldman \etal~\cite{FSS13} decomposes the stream into geometrically increasing subparts, and maintains a truncated SVD for each.  Recombining these in a careful process prevents the error from accumulating too much, and takes $O(n \cdot \s{poly}(d,\ell, \log n))$ time.

\paragraph{Catalog of Related Bounds.}
We have summarized bounds of row-update streaming in Table \ref{tbl:stream}.
The space and time bounds are given in terms of $n$ (the number of rows), $d$ (the number of columns), $k$ (the specified rank to approximate), $r$ (the rank of input matrix $A$), $\eps$ (an error parameter), and $\delta$ (the probability of failure of a randomized algorithm).  
An expresion $\tilde O(x)$ hides $\s{poly}\log (x)$ terms.  
The size is sometimes measured in terms of the number of 
rows (\#R).  Otherwise, if \#R 
is not specified the space refers the number of words in the RAM model where it is assumed $O(\log nd)$ bits fit in a single word.  

Recall that the error can be measured in several ways.  
\begin{itemize}\denselist
\item A \emph{construction} result is denoted \textsf{C}.

\item If it is not constructive, it is denoted by a \textsf{P} for projection if it uses a projection $\hat A = \pi_{B}^k(A)$.  Alternatively if the weaker form of $\hat A = \Pi_{B}^k(A)$, then this is denoted $\textsf{P}_r$ where $r$ is the rank of $B$ before the projection.  

\item In most recent results a projection error of $1+\eps$ is obtained, and is denoted $\eps$\textsf{R}.  

\item Although in some cases a weaker additive error of the form 
$\|A- \hat A\|^2_F \leq \|A - A_k\|^2_F + \eps \|A\|^2_F$ is achieved and denoted $\eps$\textsf{A}.
\\
This can sometimes also be expressed as a spectral norm of the form
$\|A - \hat A\|^2_2 \leq \|A - A_k\|^2_2 + \eps \|A\|^2_F$ (note the error term $\eps \|A\|^2_F$ still has a Frobenius norm).  This is denoted $\eps$\textsf{L}$_2$.  

\item In a few cases the error does not follow these patterns and we specially denote it.  

\item Algorithms are randomized unless it is specified.  In all tables we state bounds for a constant probability of failure.  If we want to decrease the probability of failure to some parameter $\delta$, we can generally increase the size and runtime by $O(\log(1/\delta))$.  
\end{itemize}

\begin{table}[h]
\caption{\label{tbl:stream} Streaming Algorithms}
  \begin{tabular}{|p{2.7cm}|p{3.8cm}|p{4.5cm}|p{4cm}|}
\cline{1-4}  
\multicolumn{4}{|c|}{\label{tbl:steam}\textbf{Streaming algorithms}}   \\
\cline{1-4}
\textbf{Paper} & \textbf{Space} & \textbf{Time} & \textbf{Bound} \\
     \cline{1-4} 
 DKM06\cite{drineas2006fast2} \newline LinearTimeSVD  & 
    \#R = $O(1/\eps^2)$ \newline $O((d+1/\eps^2)/\eps^4)$ & 
    $O((d+1/\eps^2)/\eps^4 + \nnz(A))$ & 
    \pbox{4cm}{\textsf{P}, $\eps$\textsf{L}$_2$} \\
    \cline{2-4}
    & \#R = $O(k/\eps^2)$ \newline $O((k/\eps^2)^2(d+k/\eps^2))$ & 
    $O((k/\eps^2)^2(d+k/\eps^2)+ \nnz(A))$ & 
    \pbox{4cm}{\textsf{P}, $\eps$\textsf{A}} \\

    \cline{1-4}
  Sar06\cite{sarlos2006improved}   \newline turnstile
    & \#R  = $O(k/\eps + k\log k)$ \newline 
        $O(d(k/\eps + k\log k))$ & 
        $O(\nnz(A)(k/\eps+k \log k) + d(k/\eps+k \log k)^2))$ & 
       \textsf{P}$_{O(k/\eps + k \log k)}$, $\eps$\textsf{R}\\

 \cline{1-4} 
 	CW09\cite{cw09}   &\#R =  $O(k/\eps)$  & $O(nd^2 + (ndk/\eps))$ & \textsf{P}$_{O(k/\eps)}$, $\eps$\textsf{R} \\

 \cline{1-4} 
 	CW09\cite{cw09} &$O((n+d)(k/\eps))$  & $O(nd^2 + (ndk/\eps))$ & \textsf{C}, $\eps$\textsf{R} \\
	
 	\cline{1-4}
    FSS13\cite{FSS13} \newline deterministic & $O((dk/\eps) \log n)$ & $n ((dk/\eps)\log n)^{O(1)}$ &
    $\textsf{P}$, $\eps \textsf{R}$ \\

 	\cline{1-4}
	\hline \hline
    \textit{This paper} \newline deterministic & \#R = $\lceil 1/\eps \rceil$ \newline $\Theta(d/\eps)$ & $O(nd/\eps)$ & $\s{covar-err} \leq \eps$ \\
 	\cline{2-4}
     & \#R = $\lceil k/\eps + k \rceil$  \newline $\Theta(dk/\eps)$ & $O(ndk/\eps)$ & \textsf{P}, $\eps$\textsf{R}\\
	\cline{1-4}    
    \end{tabular}
\end{table}

It is worth noting that under the construction model, and allowing streaming turnstile updates to each element of the matrix, the hashing algorithm has been shown space optimal by Clarkson and Woodruff~\cite{cw09} (assuming each matrix entry requires $O(\log nd)$ bits, and otherwise off by only a $O(\log nd)$ factor).  
It is randomized and it constructs a decomposition of a rank $k$ matrix $\hat A$ that satisfies $\|A - \hat A\|_F \leq (1+\eps) \|A - A_k\|_F$, with probability at least $1-\delta$.   This provides a relative error construction bound of size $O((k/\eps)(n + d) \log(nd))$ bits.  
They also show an $\Omega((k/\eps)(n+d))$ bits lower bound.  
Our paper shows that the row-update model is strictly easier with a lower upper bound.  

Although not explicitly described in their paper~\cite{cw09}, one can directly use their techniques and analysis to achieve a weak form of a non-construction projection bound.  One maintains a matrix $B = A S$ with $m = O((k/\eps) \log(1/\delta))$ columns where $S$ is a $d \times m$ matrix where each entry is chosen from $\{-1,+1\}$ at random.
Then setting $\hat A = \pi_{B}(A)$, achieves a $\s{proj-err} = 1+\eps$, however $B$ is rank $O((k/\eps) \log(1/\delta))$ and hence that is the only bound on $\hat A$ as well.

\section{Experimental Results}
We compare \FD to five different algorithms. The first two constitute brute force and na\"{\i}ve baselines, described precisely next. The other three are common algorithms that are used in practice: \s{sampling}, \s{hashing}, and \s{random-projection}, described in Section \ref{sec:related}. 
All tested methods receive the rows of an $n\times d$ matrix $A$ one by one. They are all limited in storage to an $\ell \times d$ sketch matrix $B$ and additional $o(\ell d)$ space for any auxiliary variables. This is with the exception of the brute force algorithm that requires $\Theta(d^2)$ space. For a given input matrix $A$ we compare the computational efficiency of the different methods and their resulting sketch accuracy. The computational efficiency is taken as the time required to produce $B$ from the stream of $A$'s rows. The accuracy of a sketch matrix $B$ is measured by both the covariance error and the projection error.  
Since some of the algorithms below are randomized, each algorithm was executed 5 times for each input parameter setting. The reported results are median values of these independent executions.
All methods are implemented in python with NumPy and compiled by python 2.7 compiler, and experiments are performed on a linux machine with a 6 Core Intel (R) Xeon (R) 2.4GHz CPU and 128GB RAM.

\subsection{Competing algorithms}
\noindent \s{Brute Force}: The brute force approach produces the optimal rank $\ell$ approximation of $A$. It explicitly computes the matrix $A^TA = \sum_{i=1}^{n}A_i^TA_i$ by aggregating the outer products of the rows of $A$. The final `sketch' consists of the top $\ell$ right singular vectors and values (square rooted) of $A^TA$ which are obtained by computing its SVD. The update time of Brute Force is $\Theta(d^2)$ and its space requirement is $\Theta(d^2)$.

\noindent  \s{Na\"{\i}ve}: Upon receiving a row in $A$ the na\"{\i}ve method does nothing. The sketch it returns is an all zeros $\ell$ by $d$ matrix. This baseline is important for two reasons: First, it can actually be more accurate than random methods due to under sampling scaling issues. Second, although it does not perform any computation, it does incur computation overheads such as I/O exactly like the other methods.

The other three baselines are \s{Sampling}, \s{Hashing}, and \s{Random-projection}, described in Section \ref{sec:related}.

\subsection{Datasets}
\label{sec:datasets}
We compare the performance of our algorithm on both synthetic and real datasets.
For the \s{synthetic} dataset, each row of the generated input matrices, $A$, consists of an $m$ dimensional signal and $d$ dimensional noise ($m \ll d$).
More accurately, $A = SDU + N/\zeta$. The signal coefficients matrix $S \in \mathbb{R}^{n \times m}$ is such that $S_{i,j} \sim N(0,1)$ i.i.d. 
The matrix $D \in \R^{m \times m}$ has only diagonal non-zero entered $D_{i,i} = 1 - (i - 1)/m$, which gives linearly diminishing signal singular values. The signal row space matrix $U \in \mathbb{R}^{m \times d}$ contains a random $m$ dimensional subspace in $\mathbb{R}^d$, for clarity, $UU^T = I_d$. The matrix $SDU$ is exactly rank $m$ and constitutes the signal we wish to recover.
The matrix $N \in \mathbb{R}^{n \times d}$ contributes additive Gaussian noise $N_{i,j} \sim N (0,1)$. Due to~\cite{vershynin2011spectral}, the spectral norms of $SDU$ and $N$ are expected to be the same up to some universal constant $c_1$. Experimentally, $c_1 \approx 1$. Therefore, when $\zeta \leq c_1$ we cannot expect to recover the signal because the noise spectrally dominates it. On the other hand, when $\zeta \geq c_1$ the spectral norm is dominated by the signal which is therefore recoverable. Note that the Frobenius norm of $A$ is dominated by the noise for any $\zeta \leq c_2\sqrt{d/m}$, for another constant close to $1$, $c_2$. Therefore, in the typical case where $c_1 \leq \zeta \leq c_2 \sqrt{d/m}$, the signal is recoverable by spectral methods even though the vast majority of the energy in each row is due to noise.
In our experiments, we consider $n \in \{10000,20000,\ldots,100000\}$ ($n = 10000$ as default value), $d = 1000$, $m \in \{10,20,50\}$ ($m = 10$ as default value), and $\zeta \in \{5,10,15,20\}$ ($\zeta = 10$ as default value). In projection error experiments, we recover rank $k = m$ of the dataset.

We consider the real-world dataset \s{Birds}~\cite{birds-link} in which each row represents an image of a bird, and each column a feature. This dataset has 11788 data points and 312 features, and we recover rank $k = 100$ for projection error experiments.

\vspace{2mm}
\begin{figure}[t!]
\begin{centering}
\includegraphics[width=\figsize]{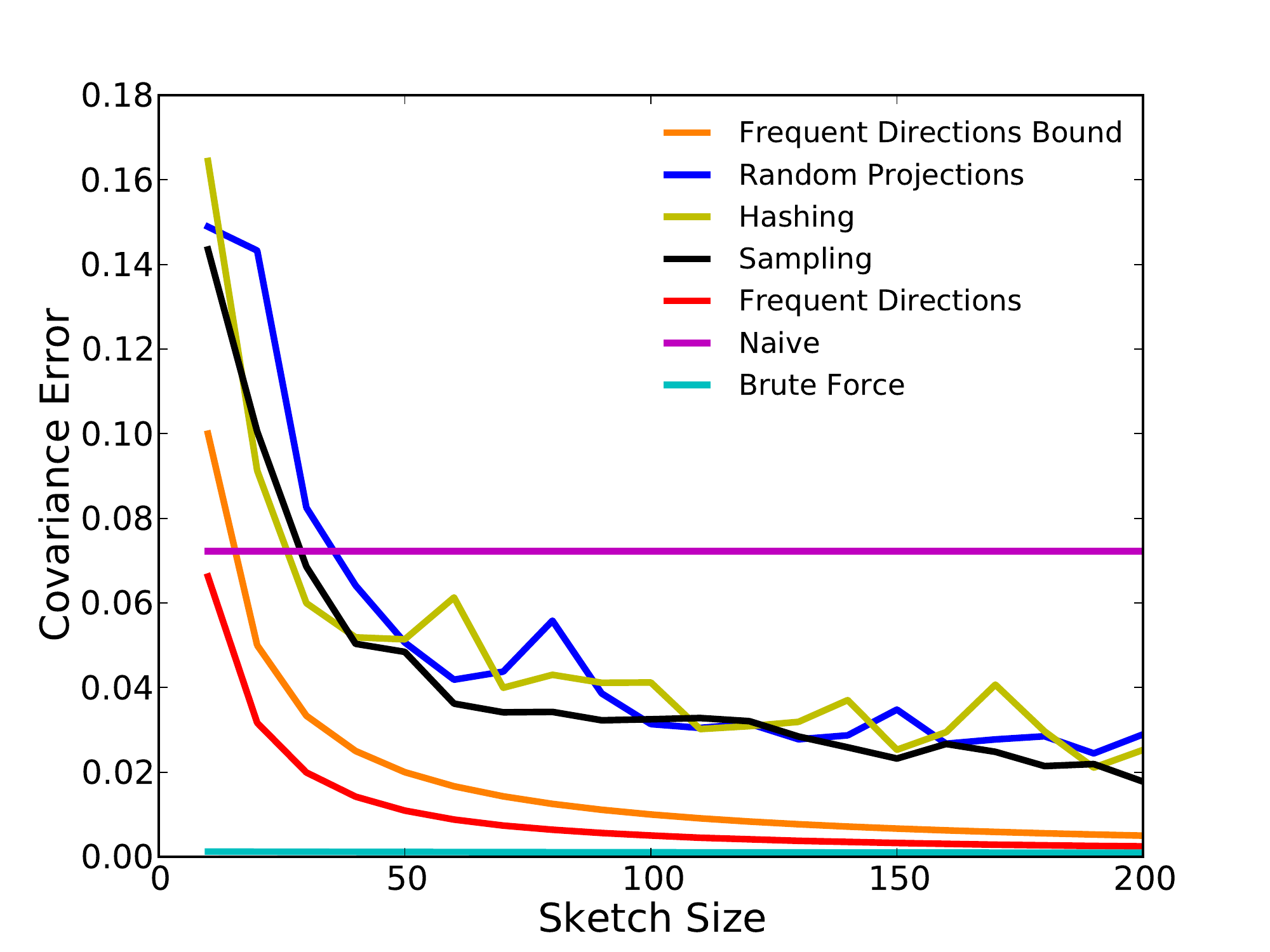}
\includegraphics[width=\figsize]{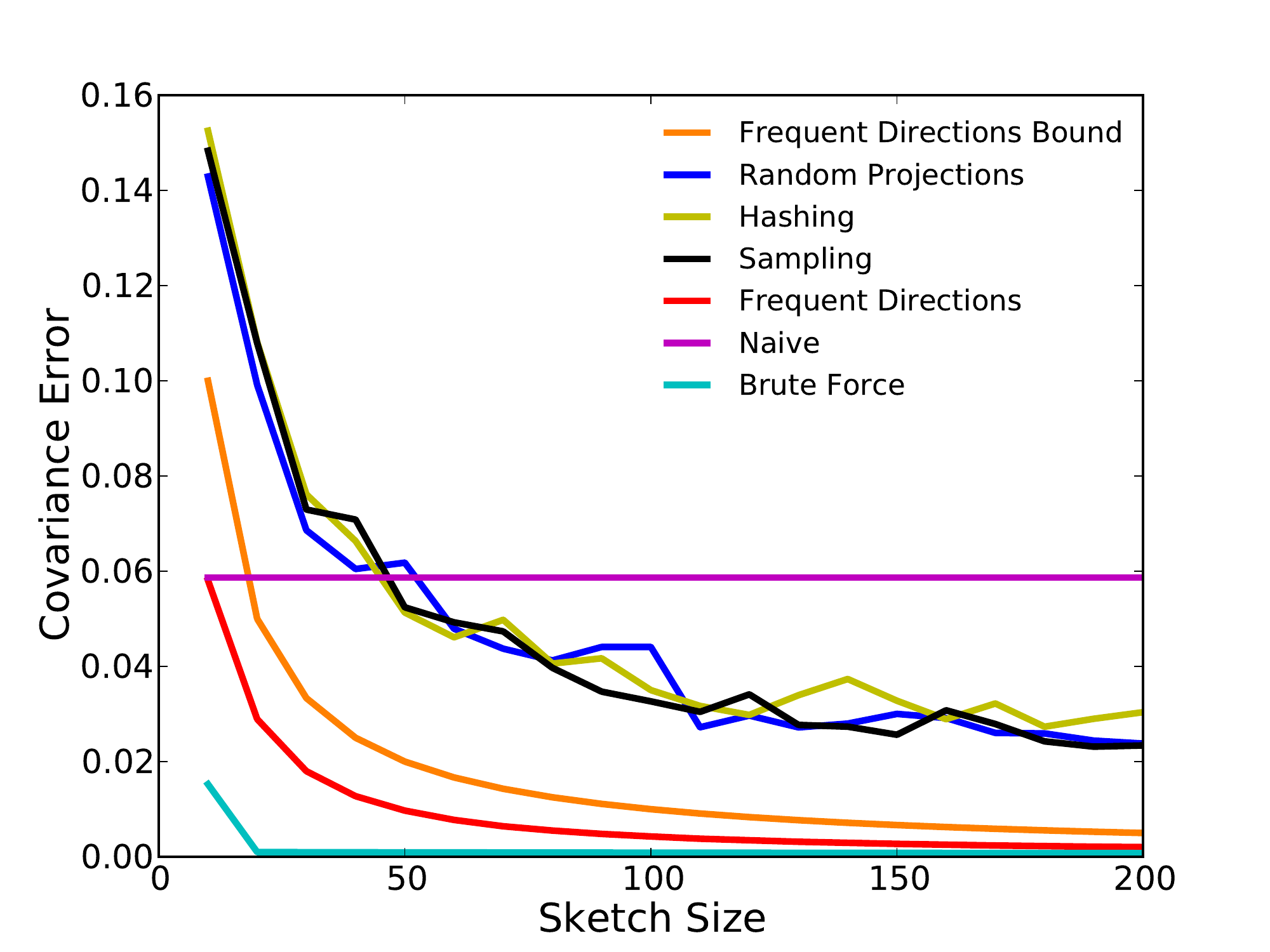}
\includegraphics[width=\figsize]{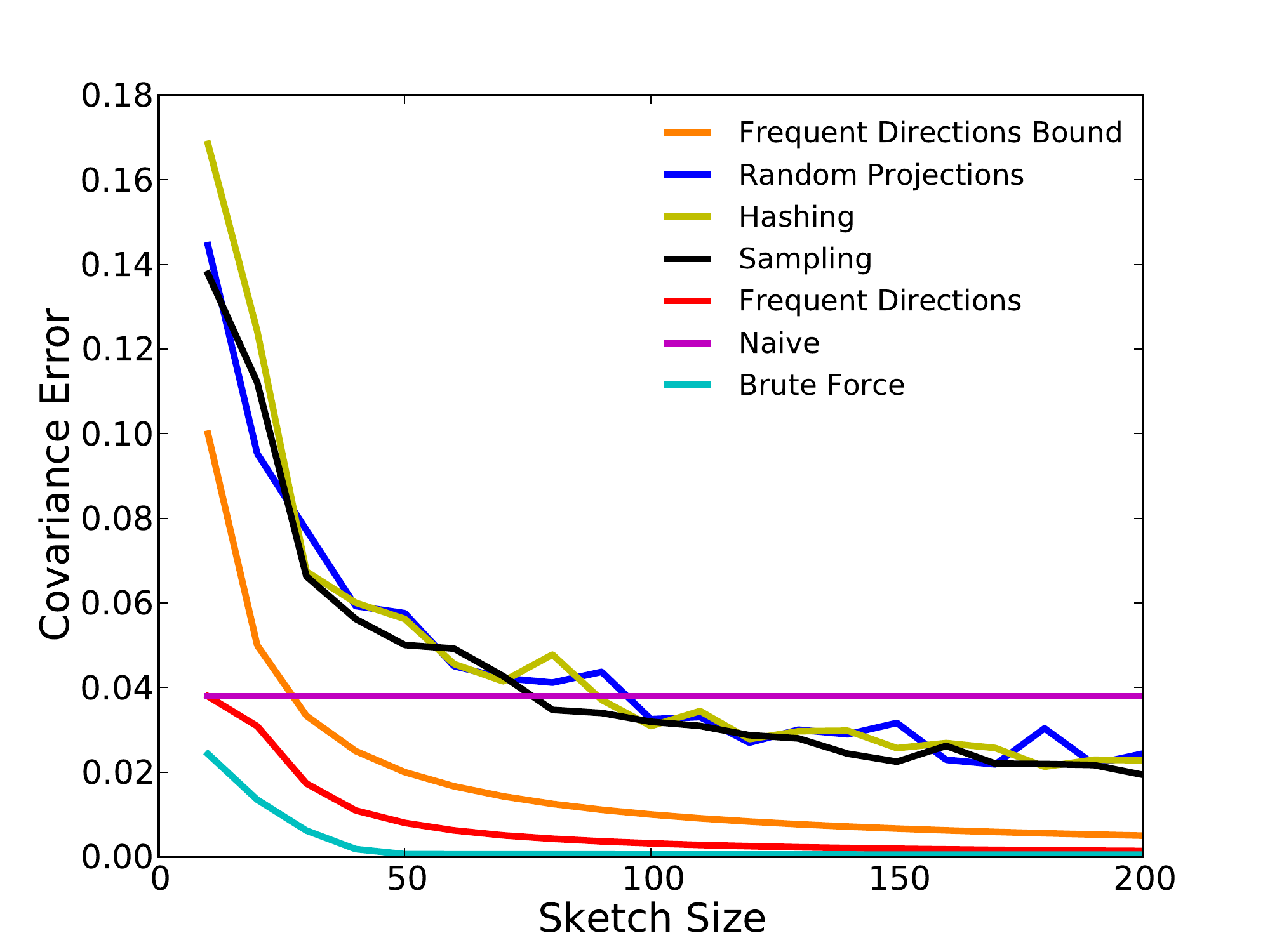}
\caption{
Covariance error vs sketch size. On synthetic data with different signal dimension ($m$). $m = 10$ (left), $m = 20$ (middle), and $m = 50$ (right).}  
\label{fig:cov_err_vs_sketch_m}
\end{centering}
\end{figure}

\begin{figure}[t!]
\begin{centering}
\includegraphics[width=\figsize]{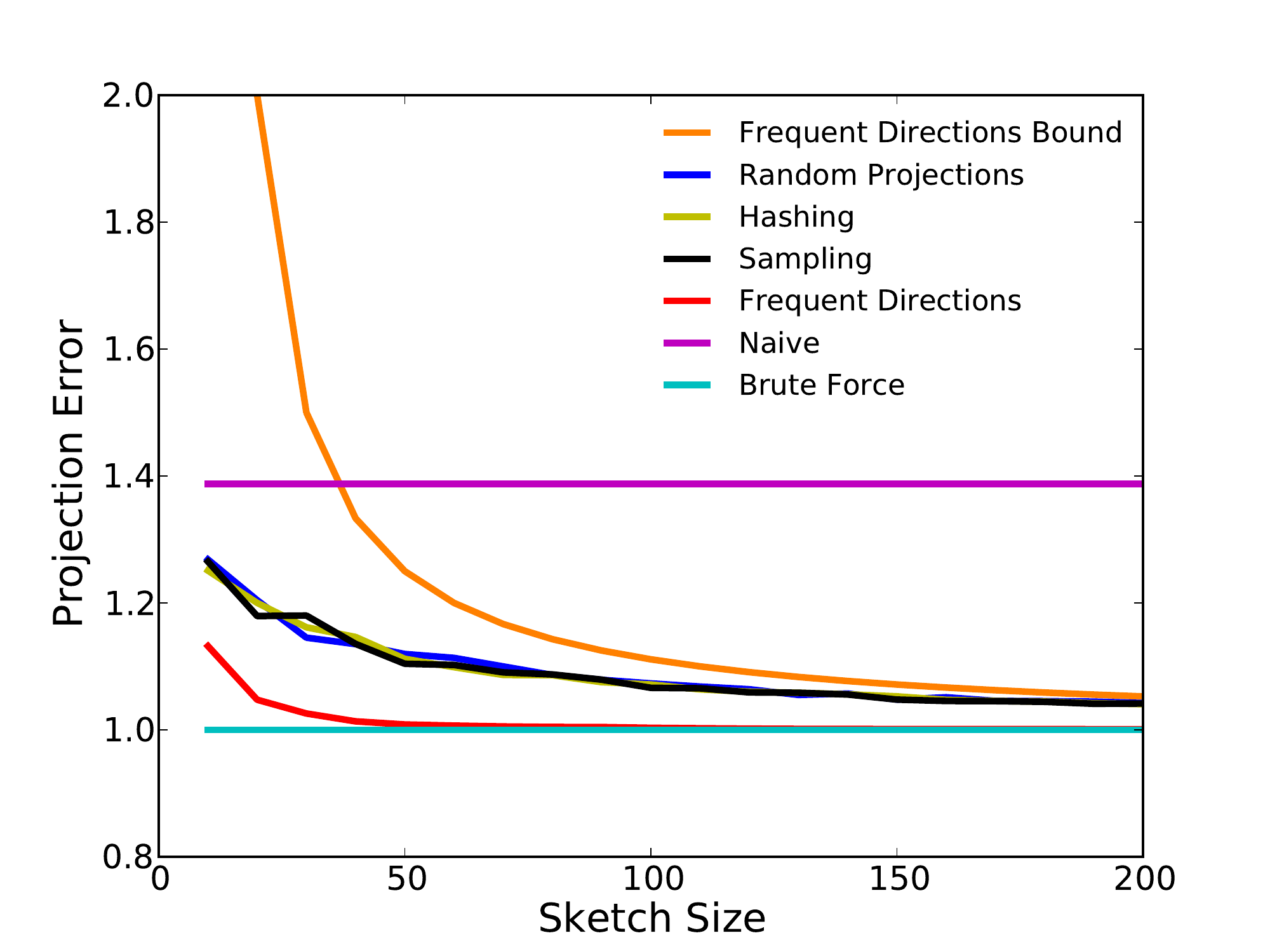}
\includegraphics[width=\figsize]{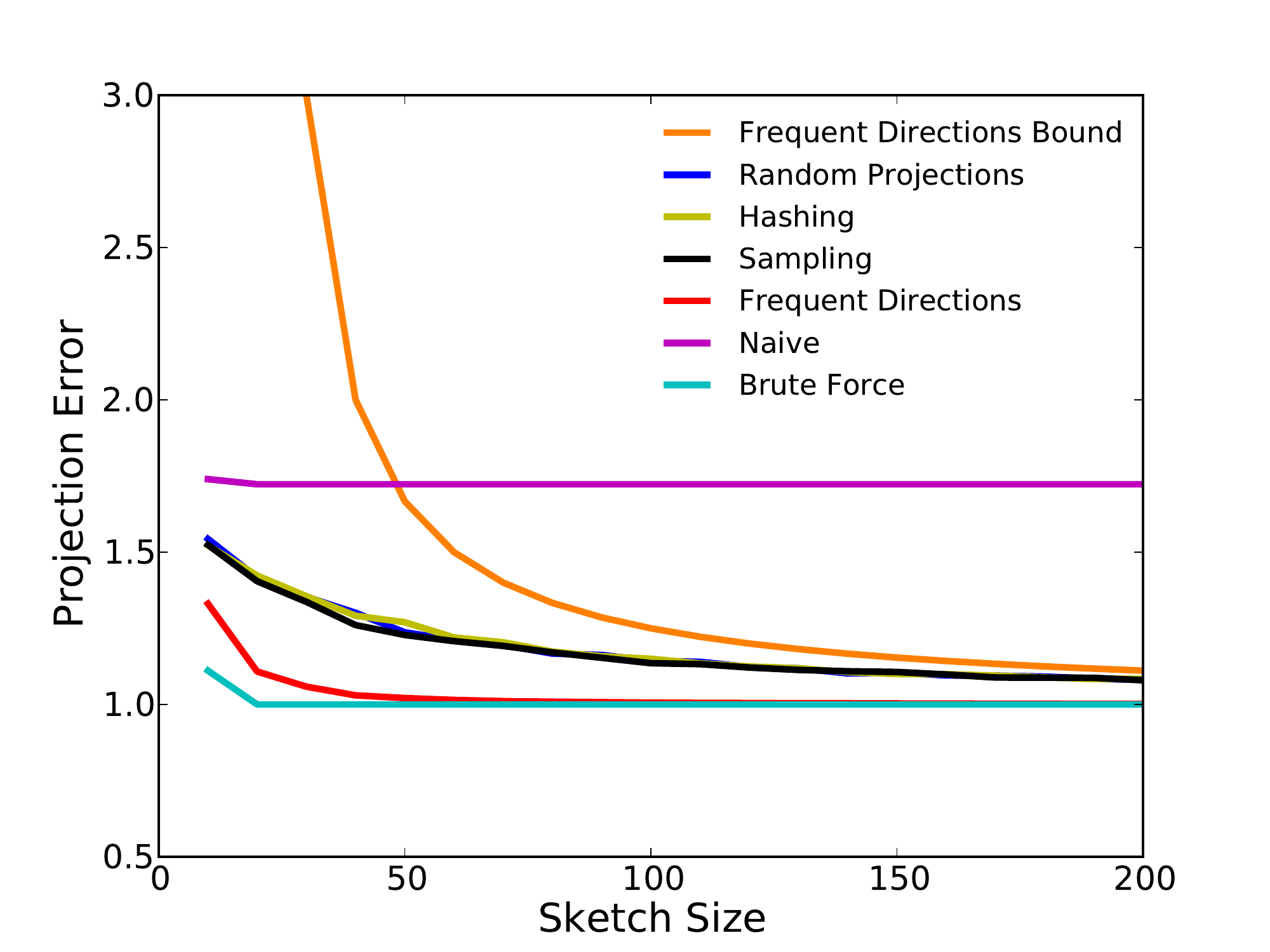}
\includegraphics[width=\figsize]{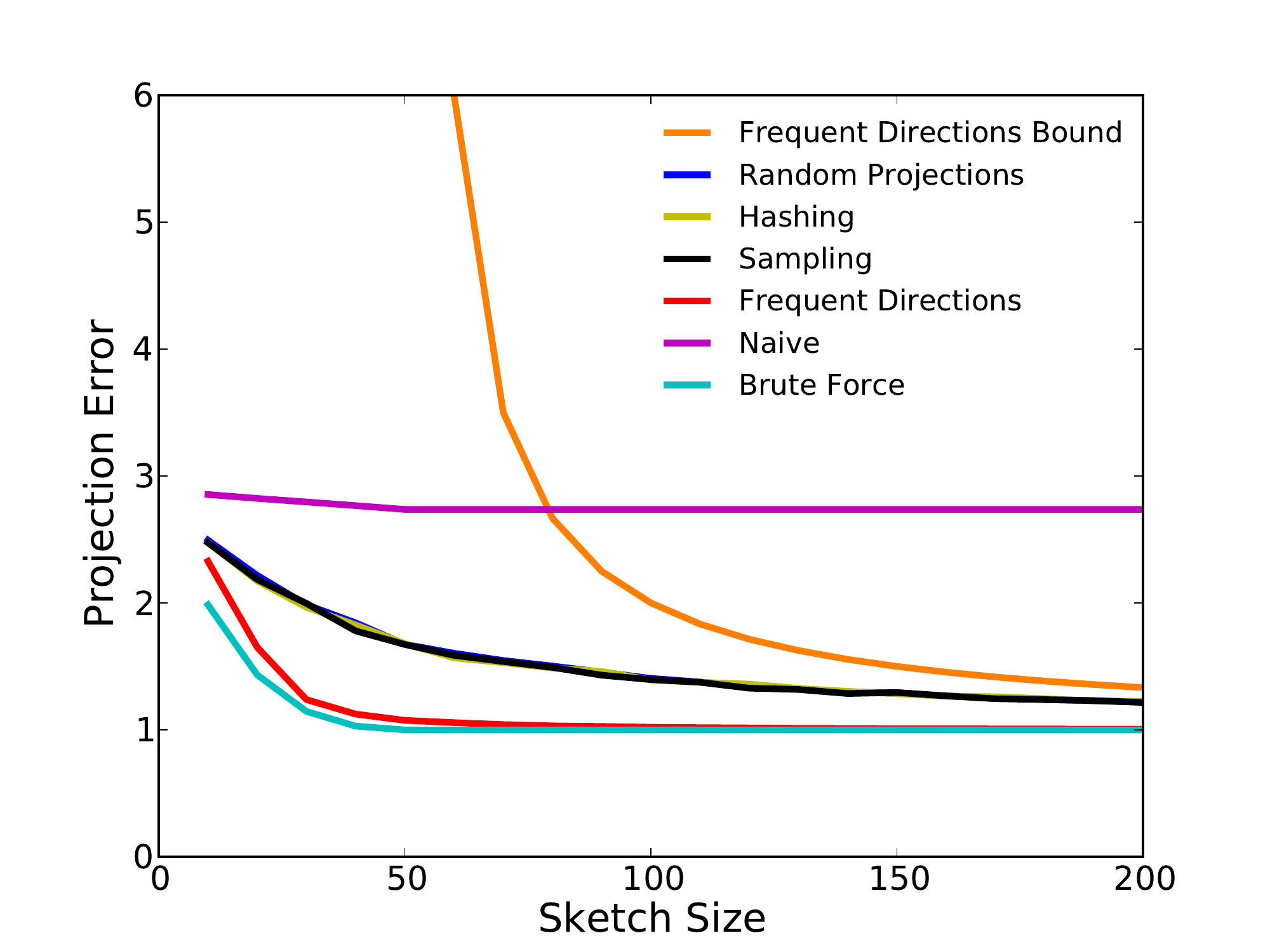}
\caption{
Projection error vs sketch size. On synthetic data with different signal dimension ($m$). $m = 10$ (left), $m = 20$ (middle), and $m = 50$ (right).}
\label{fig:prj_err_vs_sketch_m}
\end{centering}
\end{figure}

\vspace{2mm}
\begin{figure}[t!]
\begin{centering}
\includegraphics[width=\figsize]{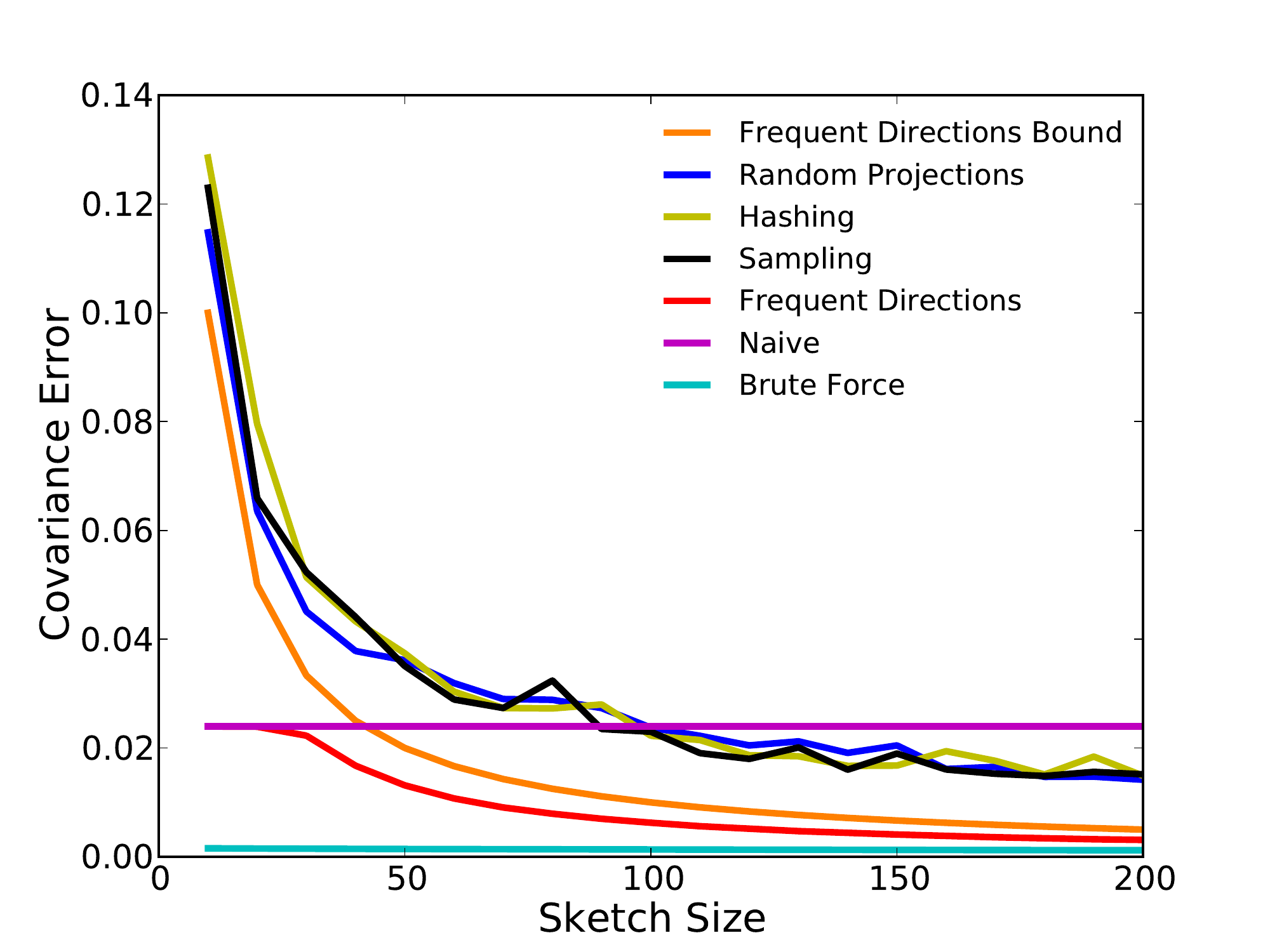}
\includegraphics[width=\figsize]{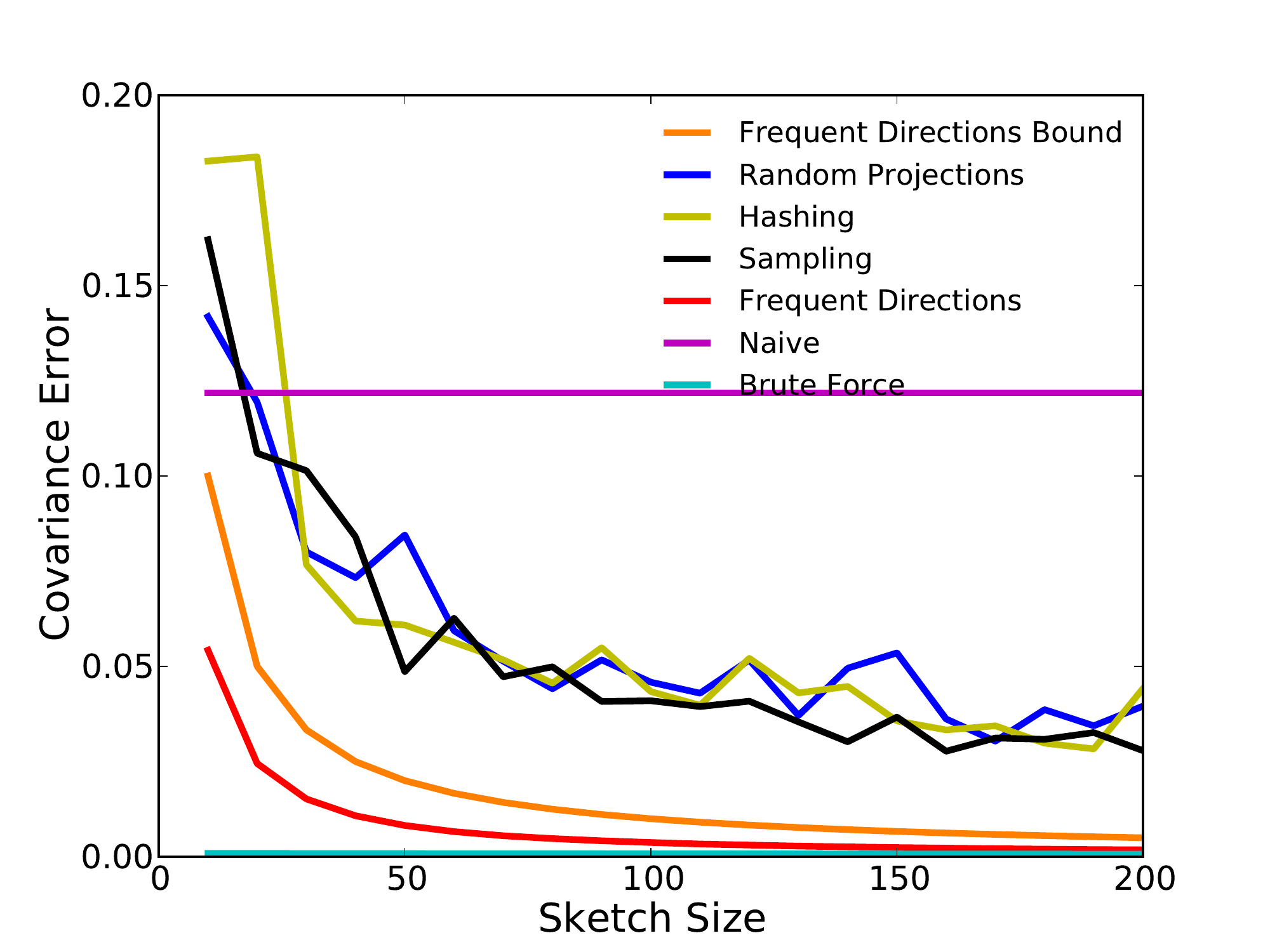}
\includegraphics[width=\figsize]{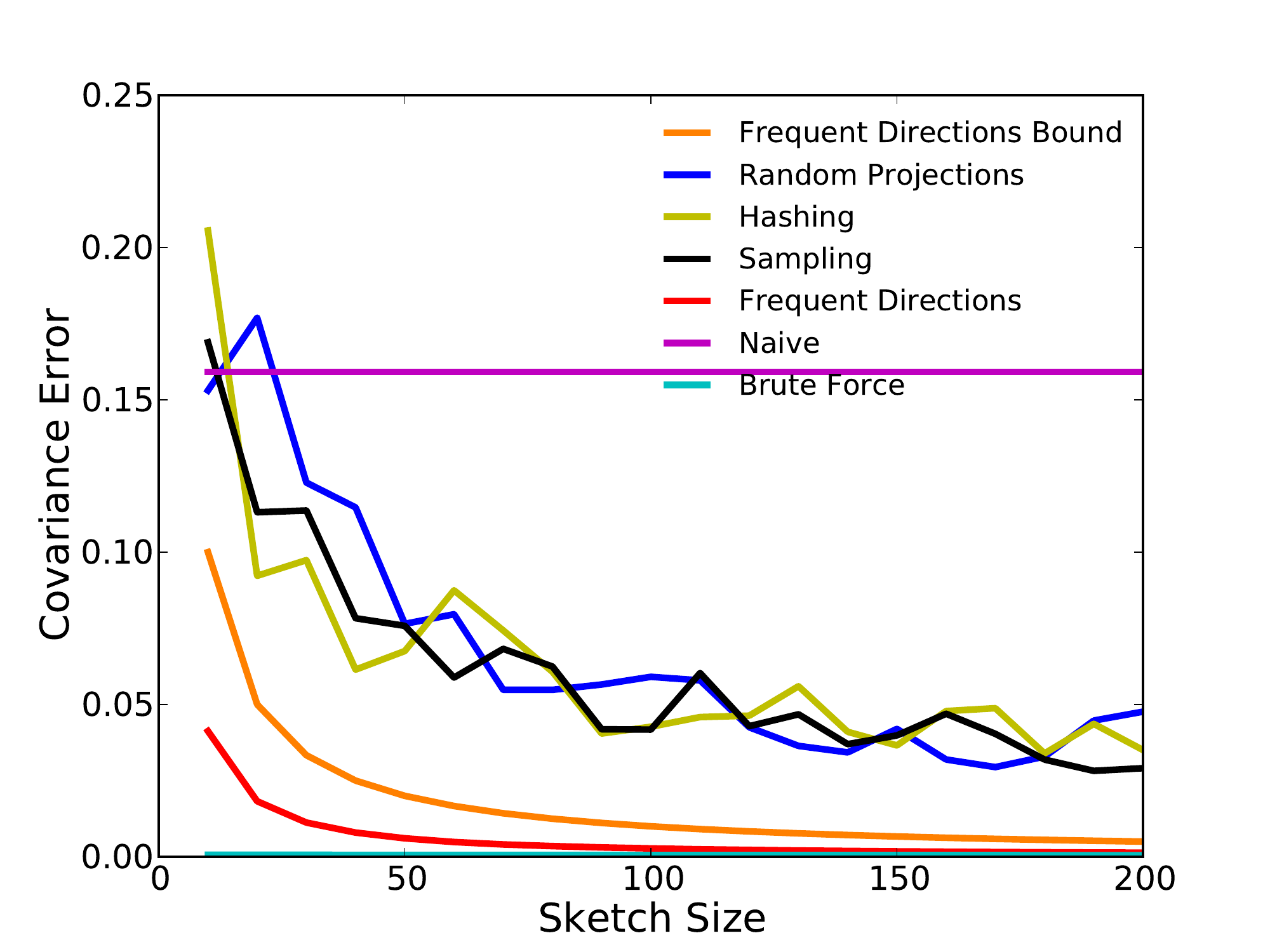}
\caption{
Covariance error vs sketch size. On synthetic data with different noise ratio ($\eta$). $\eta = 5$ (left), $\eta = 15$ (middle), and $\eta = 20$ (right).}  
\label{fig:cov_err_vs_sketch_eta}
\end{centering}
\end{figure}

\begin{figure}[t!]
\begin{centering}
\includegraphics[width=\figsize]{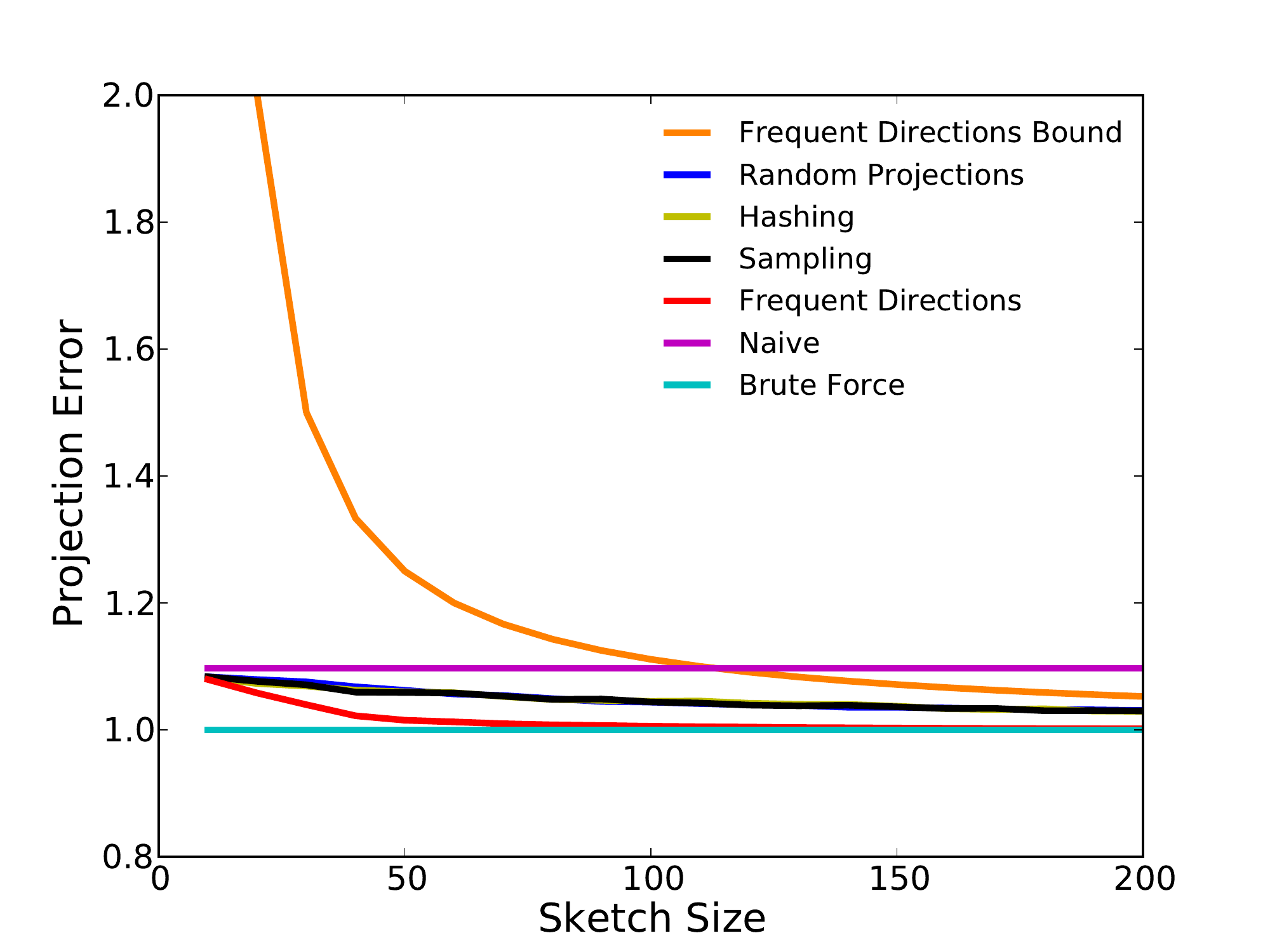}
\includegraphics[width=\figsize]{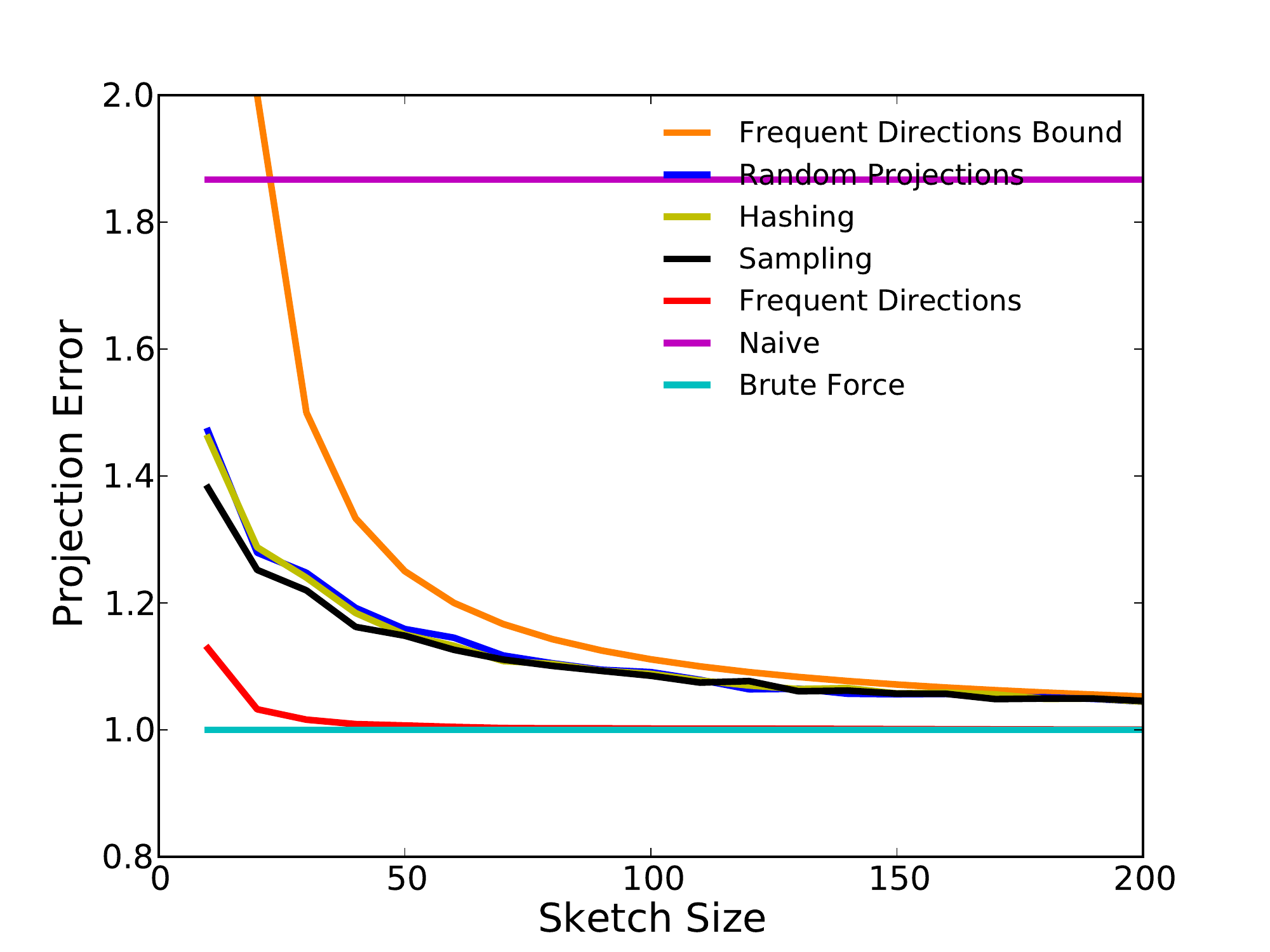}
\includegraphics[width=\figsize]{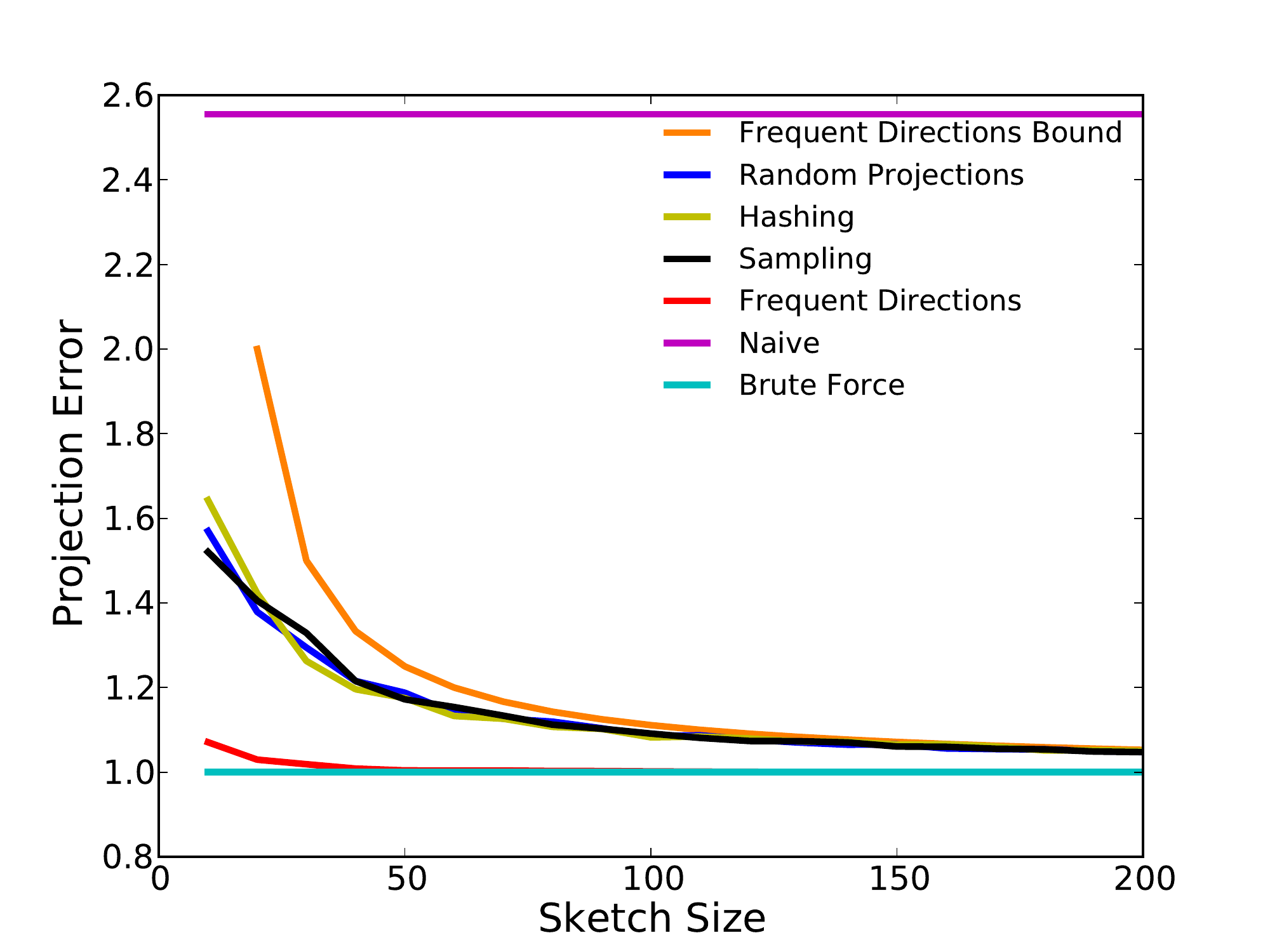}
\caption{
Projection error vs sketch size. On synthetic data with different noise ratio ($\eta$). $\eta = 5$ (left), $\eta = 15$ (middle), and $\eta = 20$ (right).}
\label{fig:prj_err_vs_sketch_eta}
\end{centering}
\end{figure}

\subsection{Results}
The performance of \FD  was measured both in terms of accuracy and running time compared to the above algorithms. In the first experiment, a moderately sized matrix ($10{,}000 \times 1{,}000$) was approximated by each algorithm. The moderate input matrix size is needed to accommodate the brute force algorithm and to enable the exact error measure. The results are shown as we vary the data dimension $m$ in Figure \ref{fig:cov_err_vs_sketch_m} for covariance error and Figure \ref{fig:prj_err_vs_sketch_m} for projection error; similar results are shown as the noise parameter $\eta$ is changed in Figure \ref{fig:cov_err_vs_sketch_eta} for covariance error and Figure \ref{fig:prj_err_vs_sketch_eta} for projection error.  
These give rise to a few interesting observations. First, all three random techniques actually perform worse in covariance error than \s{na\"{\i}ve} for small sketch sizes. This is a side effect of under-sampling which causes overcorrection. This is not the case with \FD. 
Second, the three random techniques perform equally well. This might be a result of the chosen input. Nevertheless, practitioners should consider these as potentially comparable alternatives. 
Third, the curve indicated by ``Frequent Direction Bound" plots the relative accuracy guaranteed by \FD, which is equal to $1/\ell$ in case of covariance error, and equal to $\ell/(\ell - k)$ in case of projection error. Note that in covariance error plots,``Frequent Direction Bound" is consistently lower than the random methods. This means that the worst case performance guarantee is lower than the actual performance of the competing algorithms. Finally, \FD produces significantly more accurate sketches than predicted by its worst case analysis.

The running time of \FD (specifically Algorithm \ref{alg:freqDir2}), however, is not better than its competitors. This is clearly predicted by their asymptotic running times. In Figure \ref{fig:run_time_vs_sketch_m} as we vary the data dimension $m$ and in Figure \ref{fig:run_time_vs_sketch_eta} as we vary the noise parameter $\eta$, the running times (in seconds) of the sketching algorithms are plotted as a function of their sketch sizes. Clearly, the larger the sketch, the higher the running time. Note that \s{hashing} is extremely fast. In fact, it is almost as fast as \s{na\"{\i}ve}, which does nothing other than read the data!  \s{Sampling} is also faster than Frequent Directions.  
\s{Random-project} is also faster than \FD, but only by a factor $3$; this makes sense since they have the same asymptotic runtime, but \s{random-projection} just needs to add each new row to $\ell$ other vectors while \FD amortizes over more complex \svd\ calls.  
It is important to stress that the implementations above are not carefully optimized.
Different implementations might lead to different results.

\begin{figure}[t!]
\begin{centering}
\includegraphics[width=\figsize]{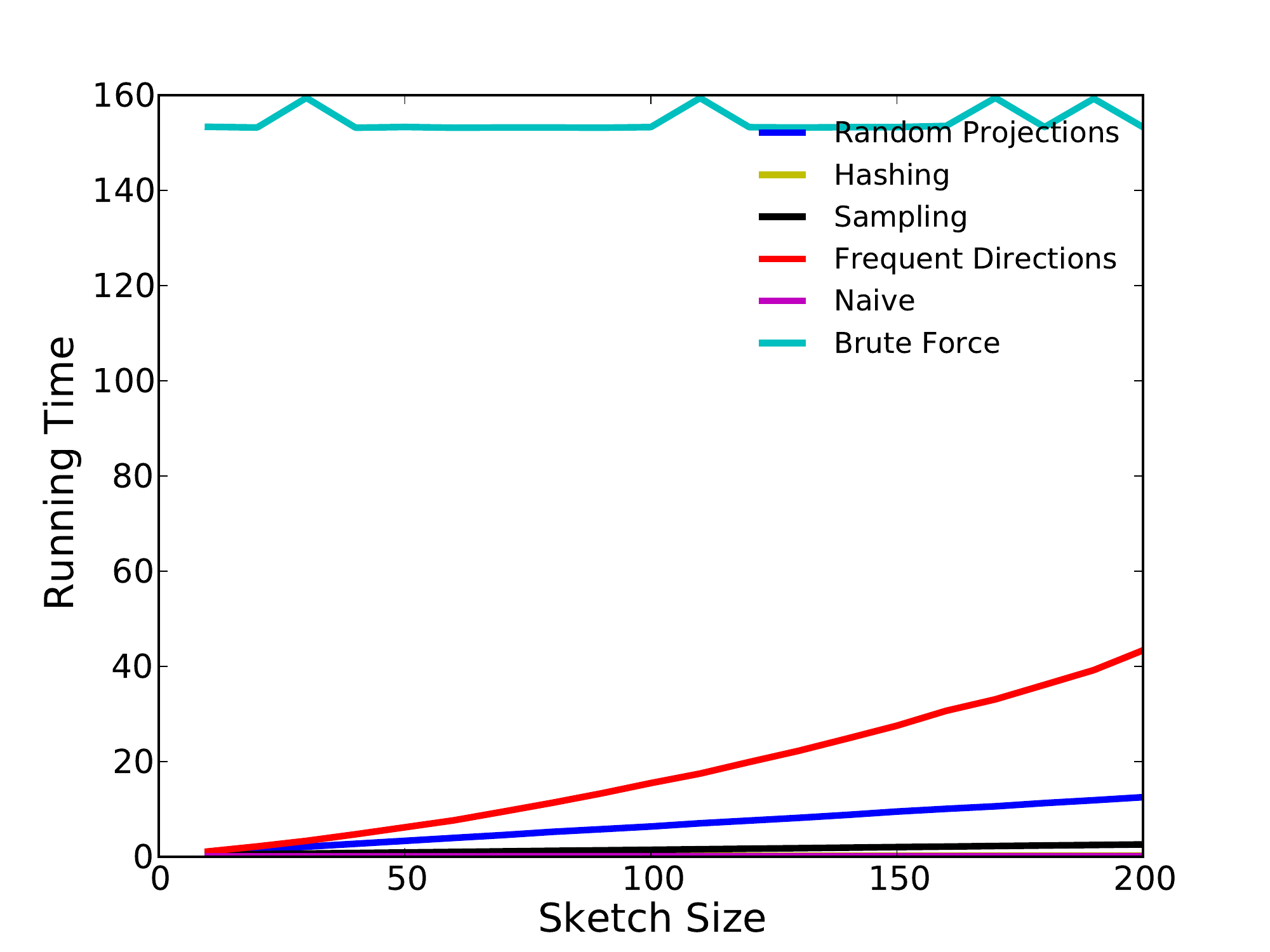}
\includegraphics[width=\figsize]{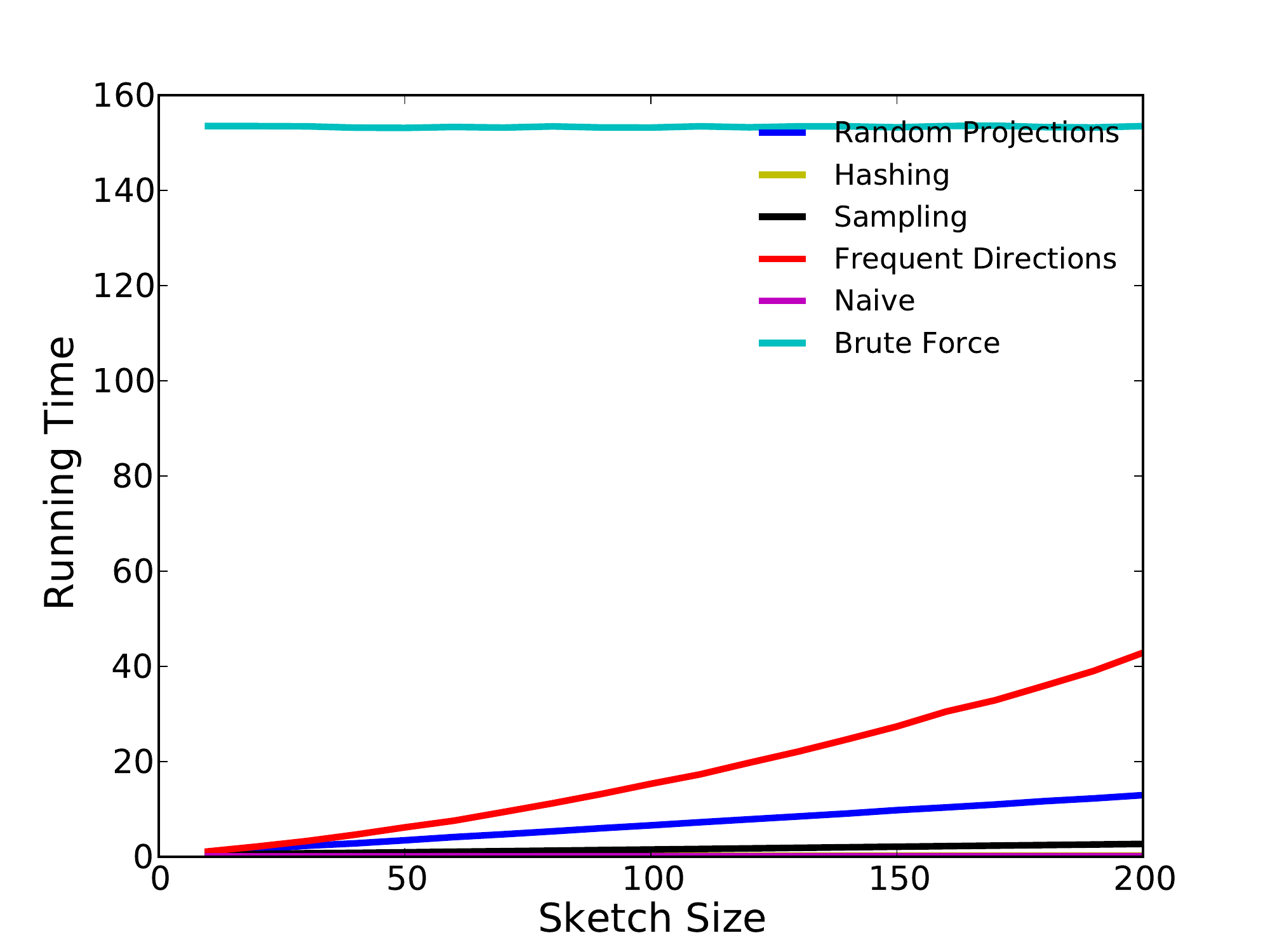}
\includegraphics[width=\figsize]{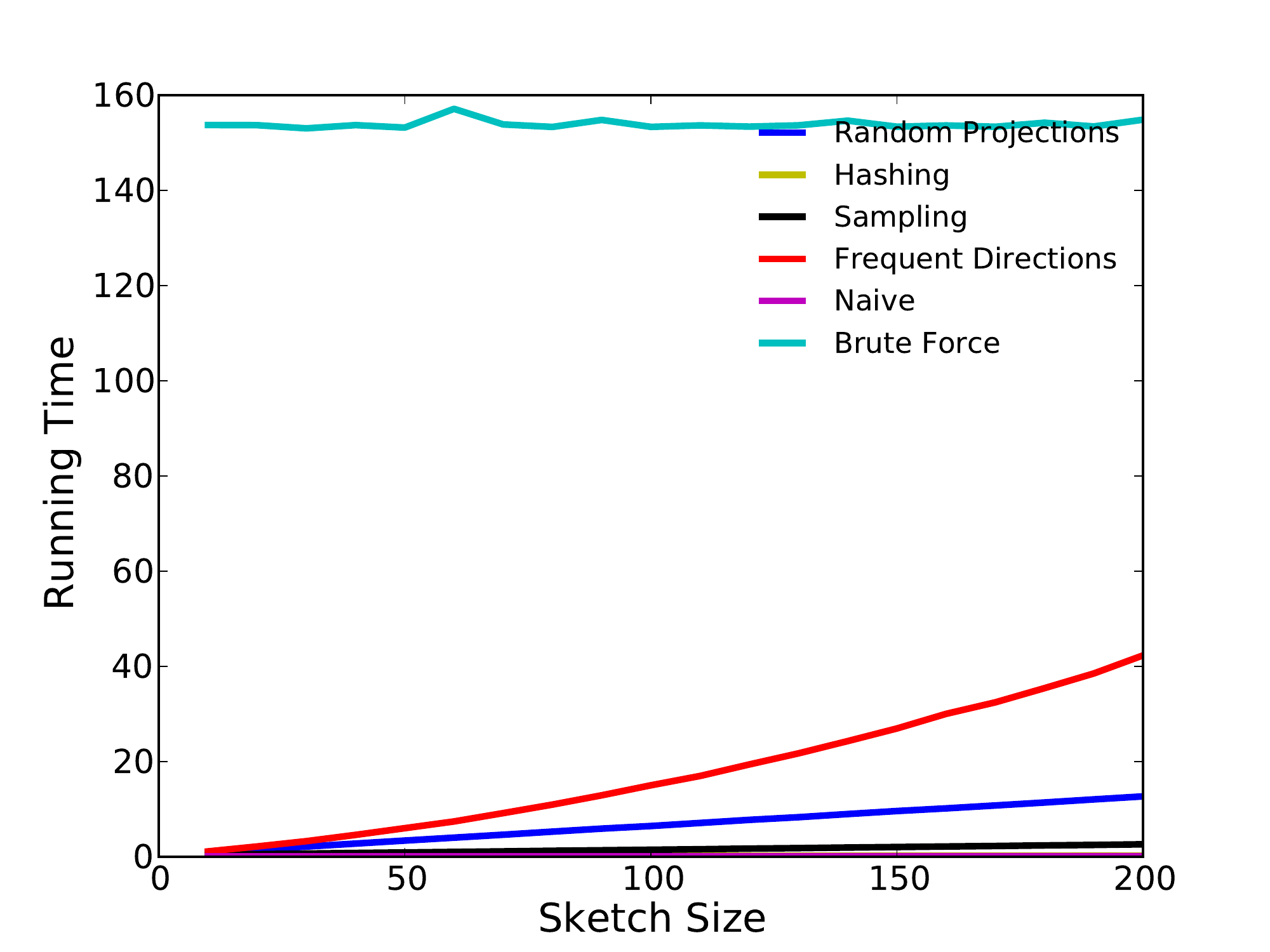}
\caption{
Running time (seconds) vs sketch size. On synthetic data with different signal dimension ($m$). $m = 10$ (left), $m = 20$ (middle), and $m = 50$ (right).}
\label{fig:run_time_vs_sketch_m}
\end{centering}
\end{figure}

\begin{figure}[t!]
\begin{centering}
\includegraphics[width=\figsize]{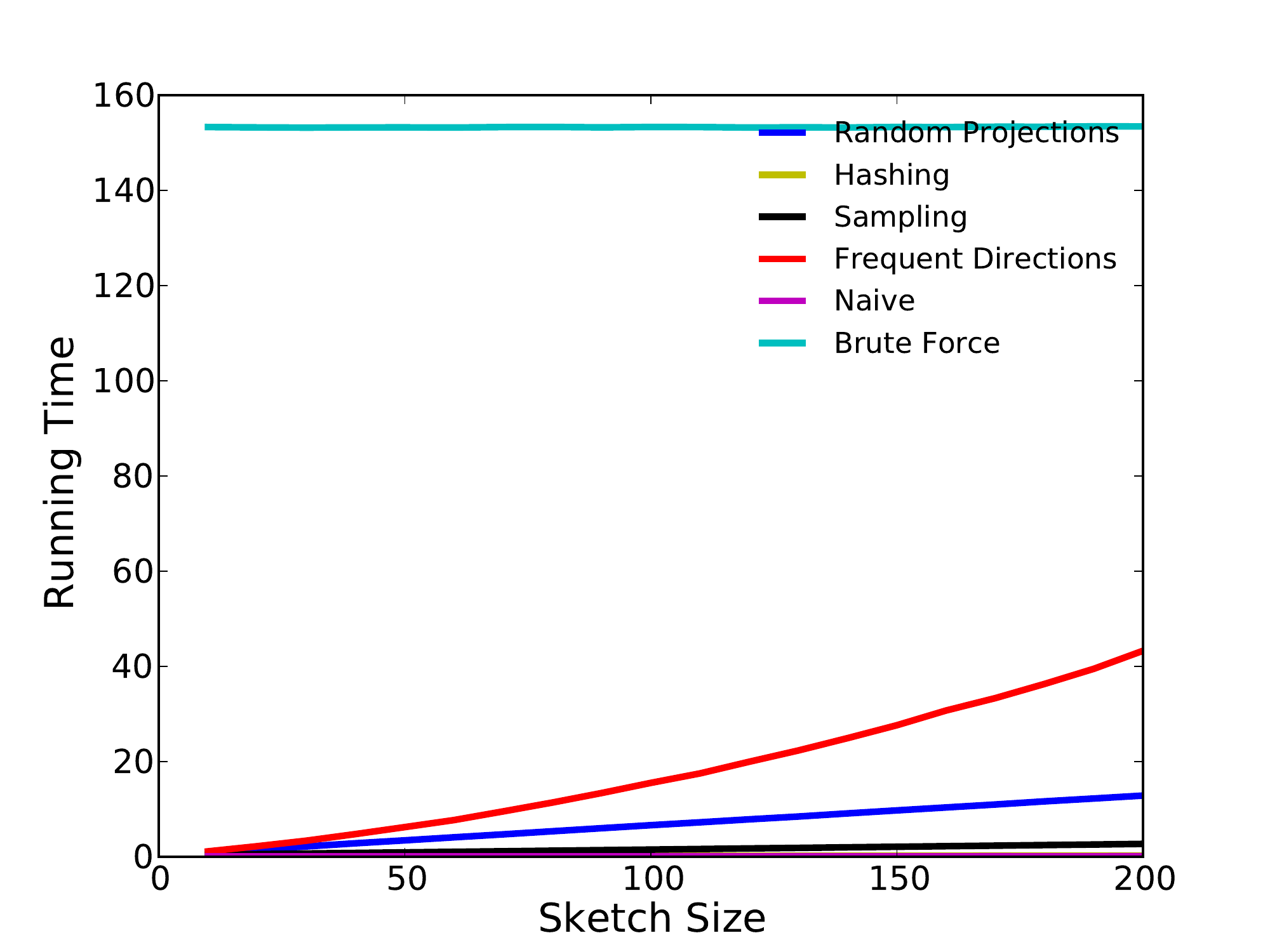}
\includegraphics[width=\figsize]{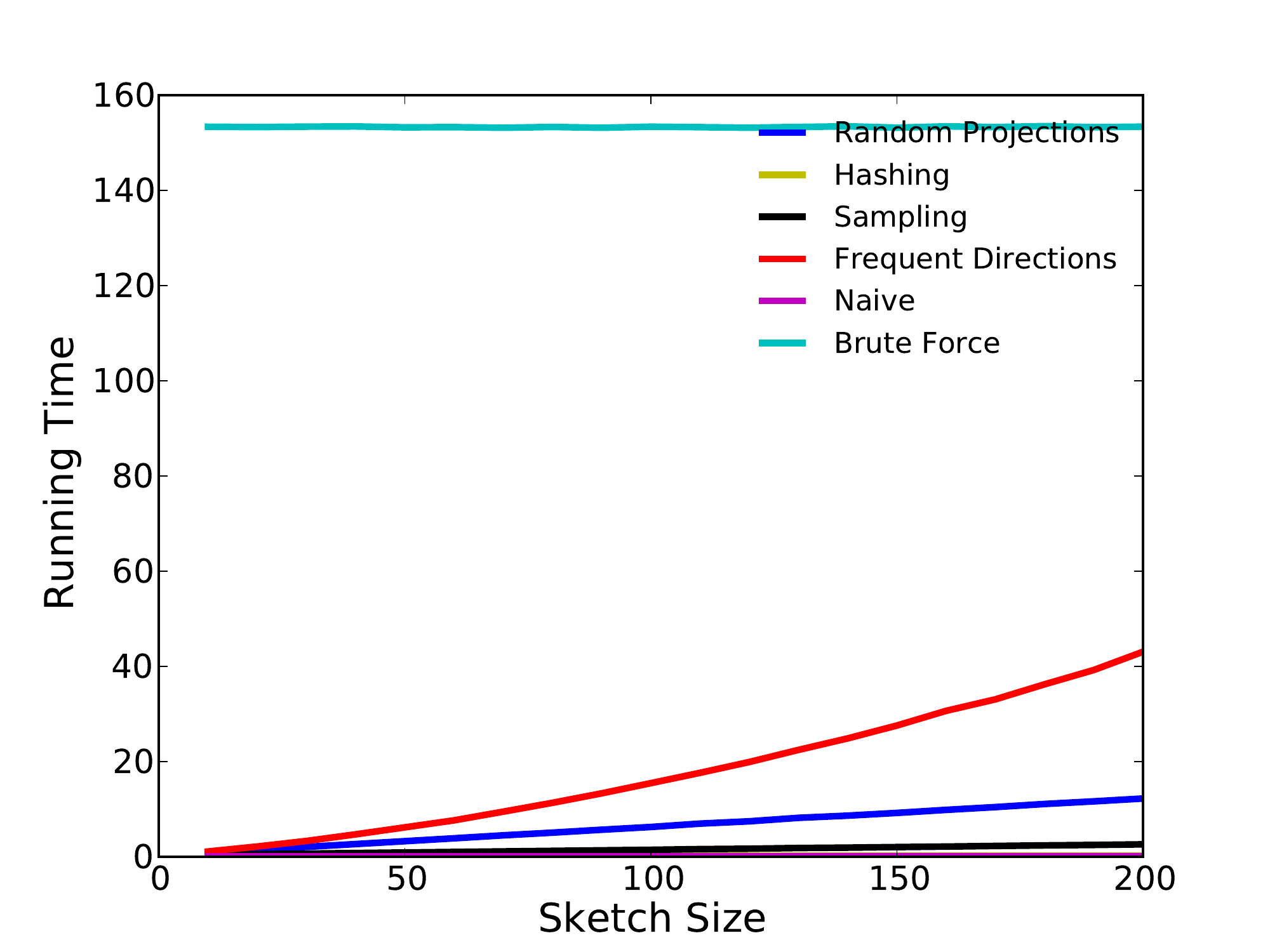}
\includegraphics[width=\figsize]{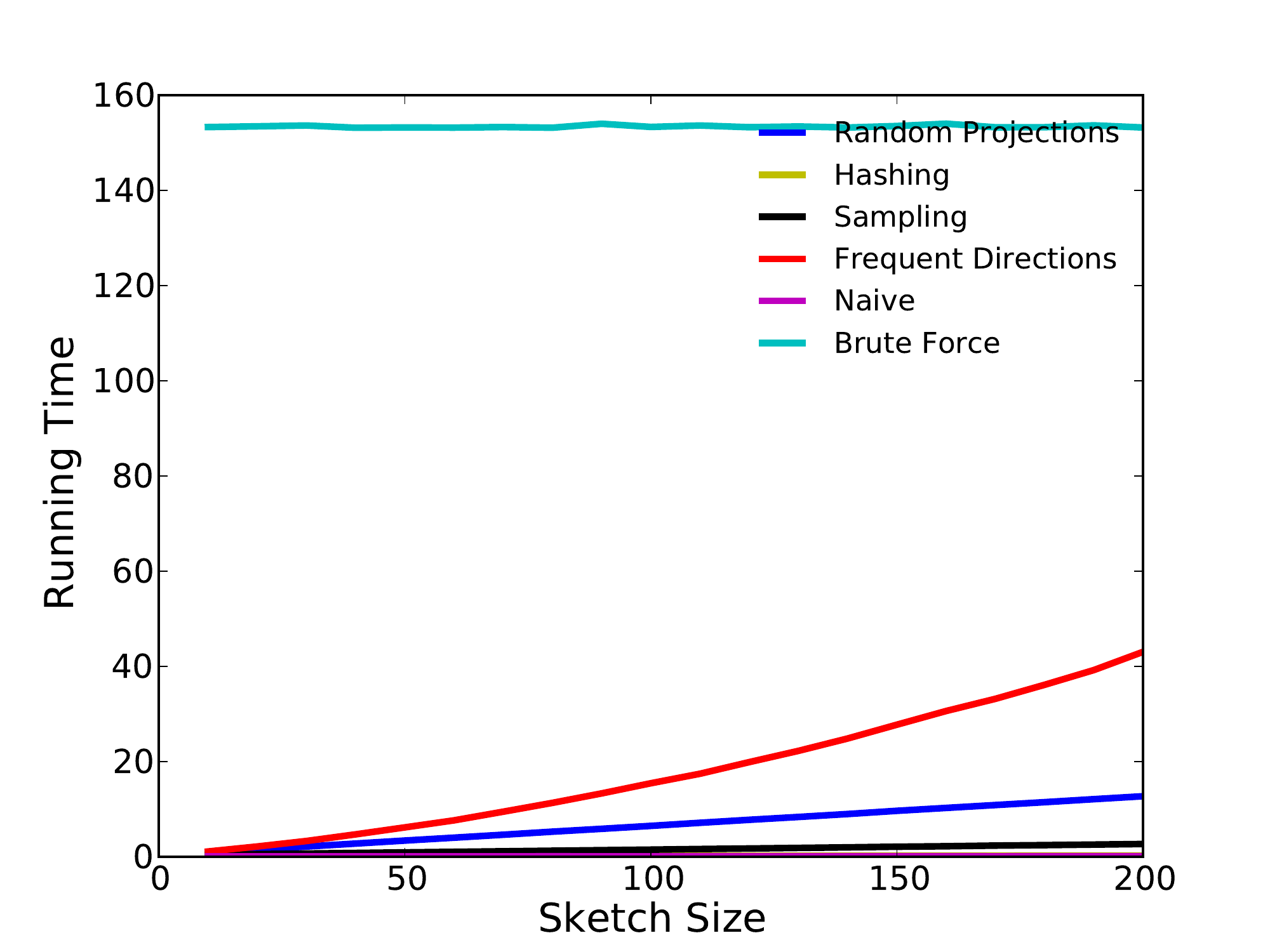}
\caption{
Running time (seconds) vs sketch size. On synthetic data with different noise ratio ($\eta$). $\eta = 5$ (left), $\eta = 10$ (middle), and $\eta = 15$ (right).}
\label{fig:run_time_vs_sketch_eta}
\end{centering}
\end{figure}

Similar results in error and runtime are shown on the real \s{Birds} data set in Figure \ref{fig:birds}.  

\begin{figure}[t!]
\begin{centering}
\includegraphics[width=\figsize]{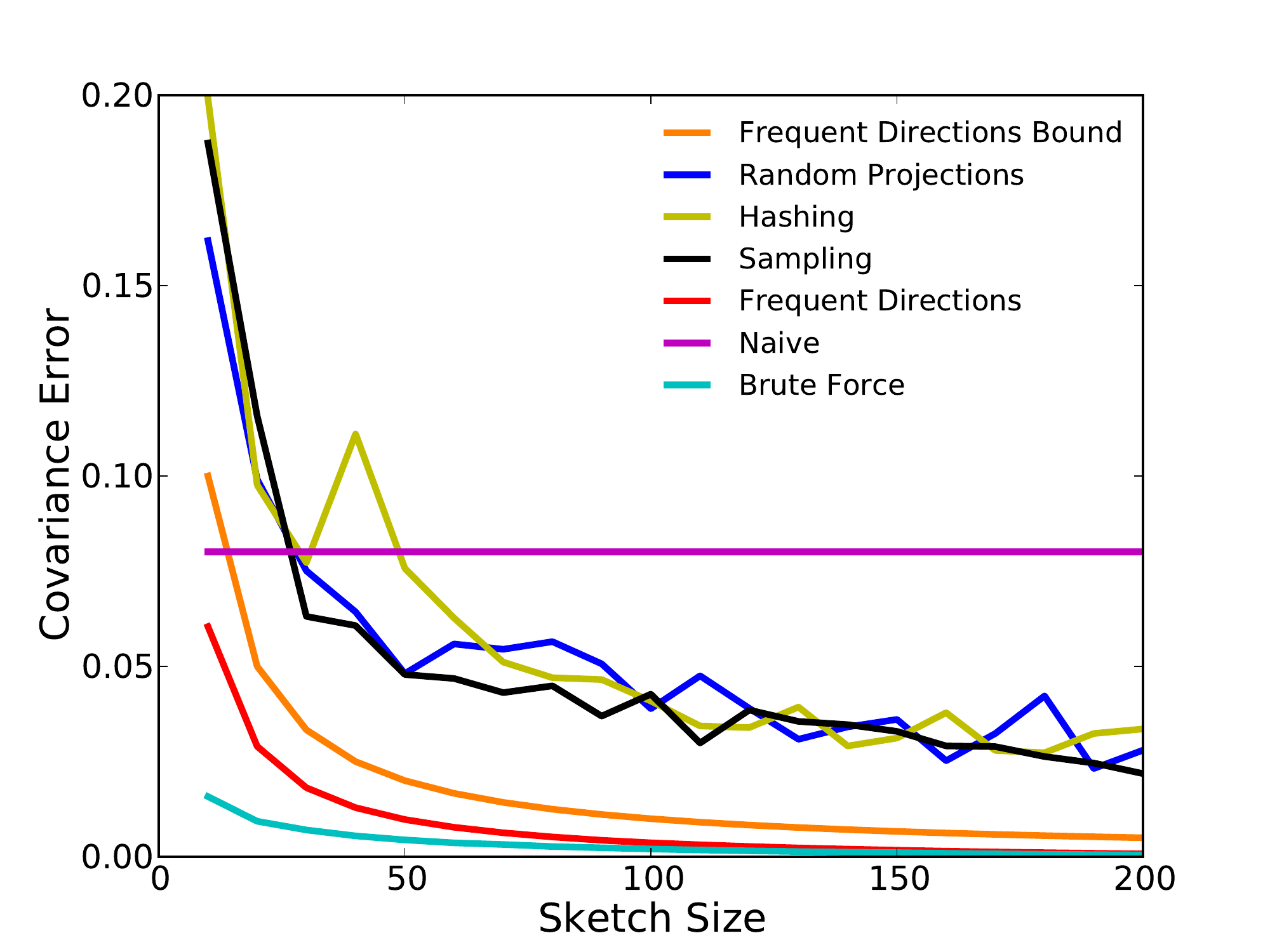}
\includegraphics[width=\figsize]{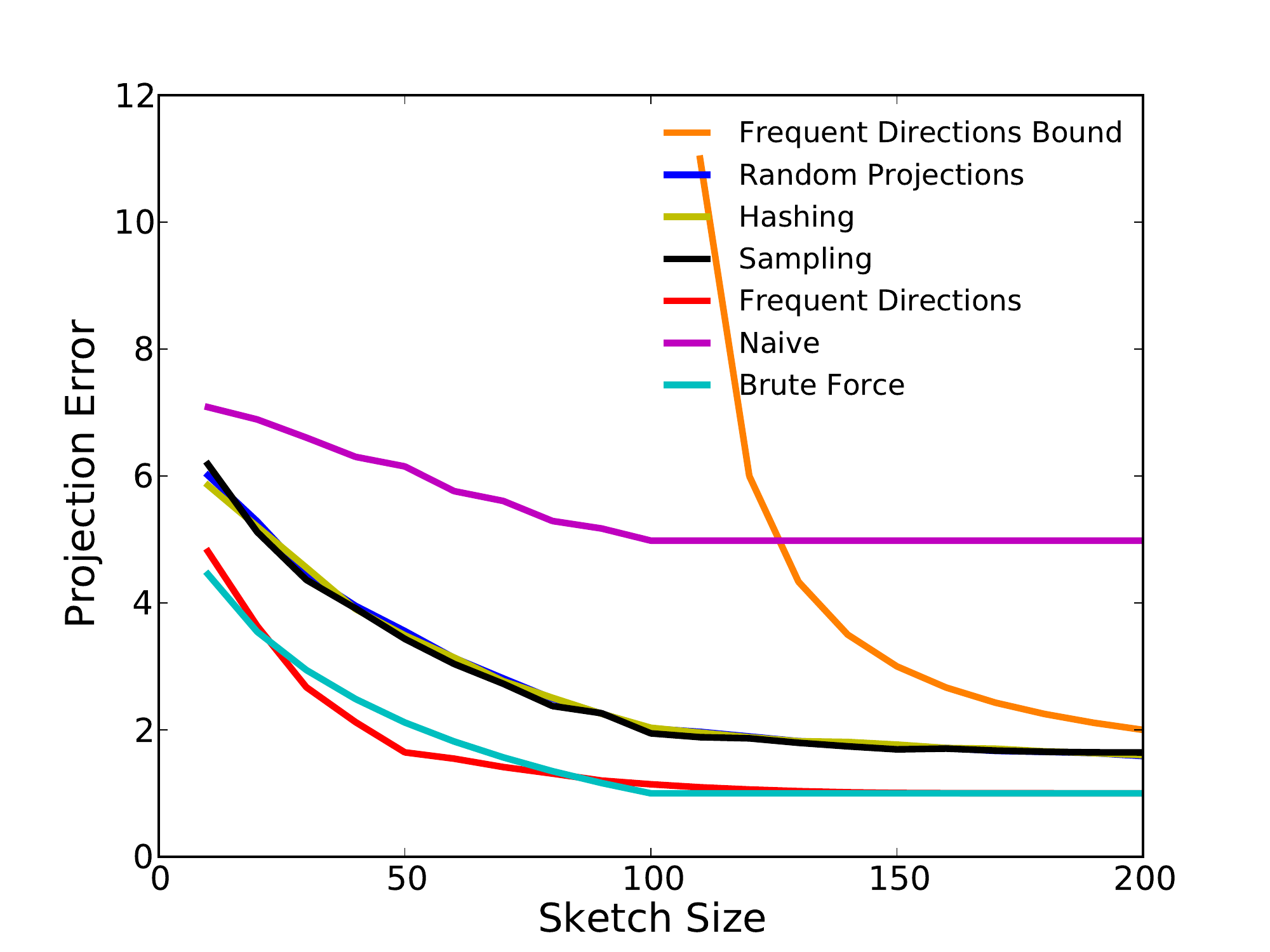}
\includegraphics[width=\figsize]{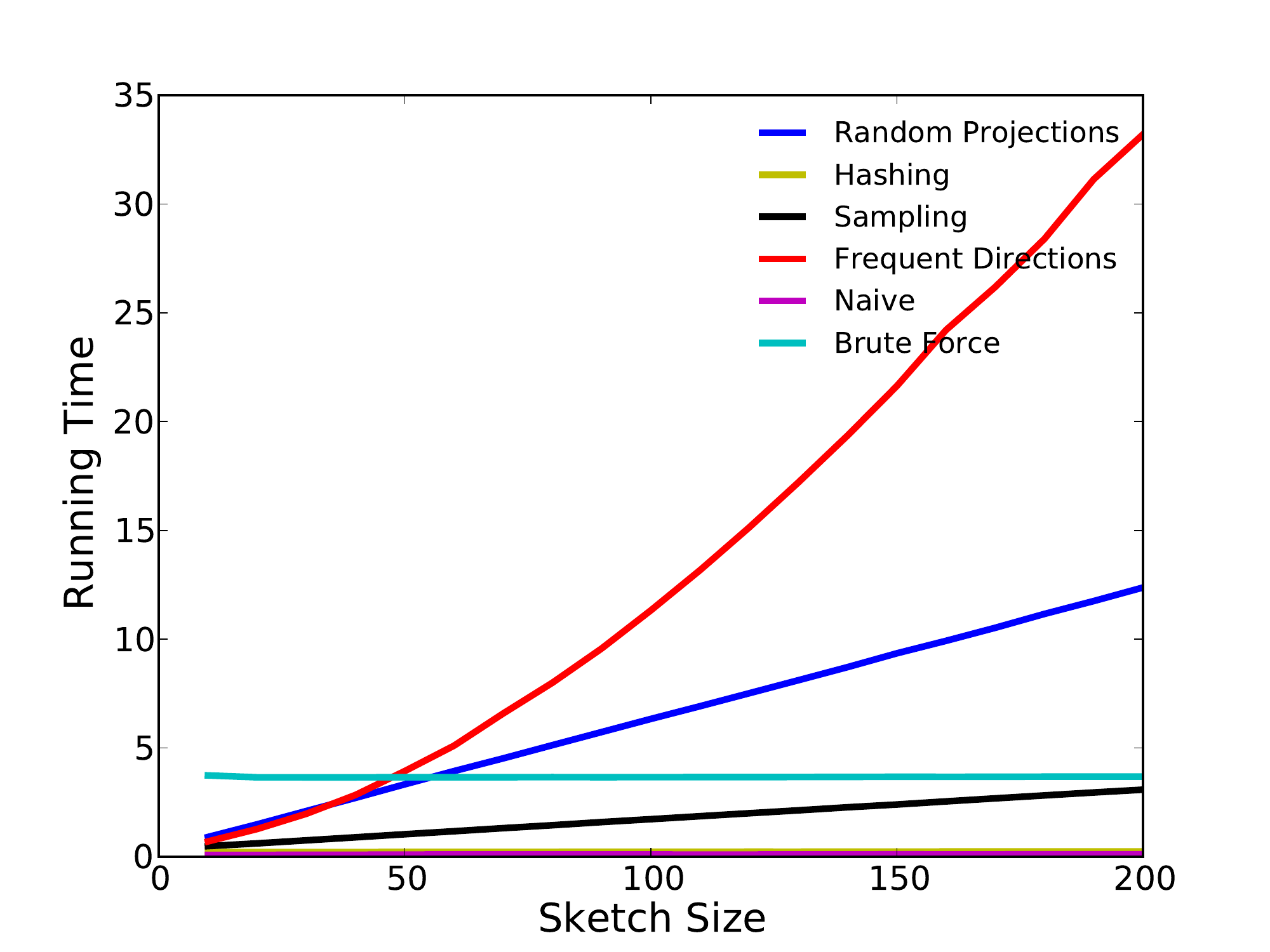}
\caption{
Covariance error (left), Projection error (middle) and Running time (right) on birds data.}
\label{fig:birds}
\end{centering}
\end{figure}

Nevertheless, we will claim that \FD scales well. Its running time is $O(nd\ell)$, which is linear in each of the three terms. In Figure \ref{fig:time_vs_input}, we fix the sketch size to be $\ell = 100$ and increase $n$ and $d$. Note that the running time is indeed linear in both $n$ and $d$ as predicted. Moreover, sketching an input matrix of size $10^5 \times 10^4$ requires roughly 3 minutes. Assuming 4 byte floating point numbers, this matrix occupies roughly 4Gb of disk space. More
importantly though, \FD is a streaming algorithm. Thus, its memory footprint is fixed and its running time is exactly linear in $n$. For example, sketching a 40Gb matrix of size $10^6 \times 10^4$
terminates in half an hour. The fact that \FD is also perfectly parallelizable (Section \ref{sec:parallelization}) makes \FD applicable to truly massive and distributed matrices.

\begin{figure}[t!]
\begin{centering}
\subfigure{
\includegraphics[width=\figsize]{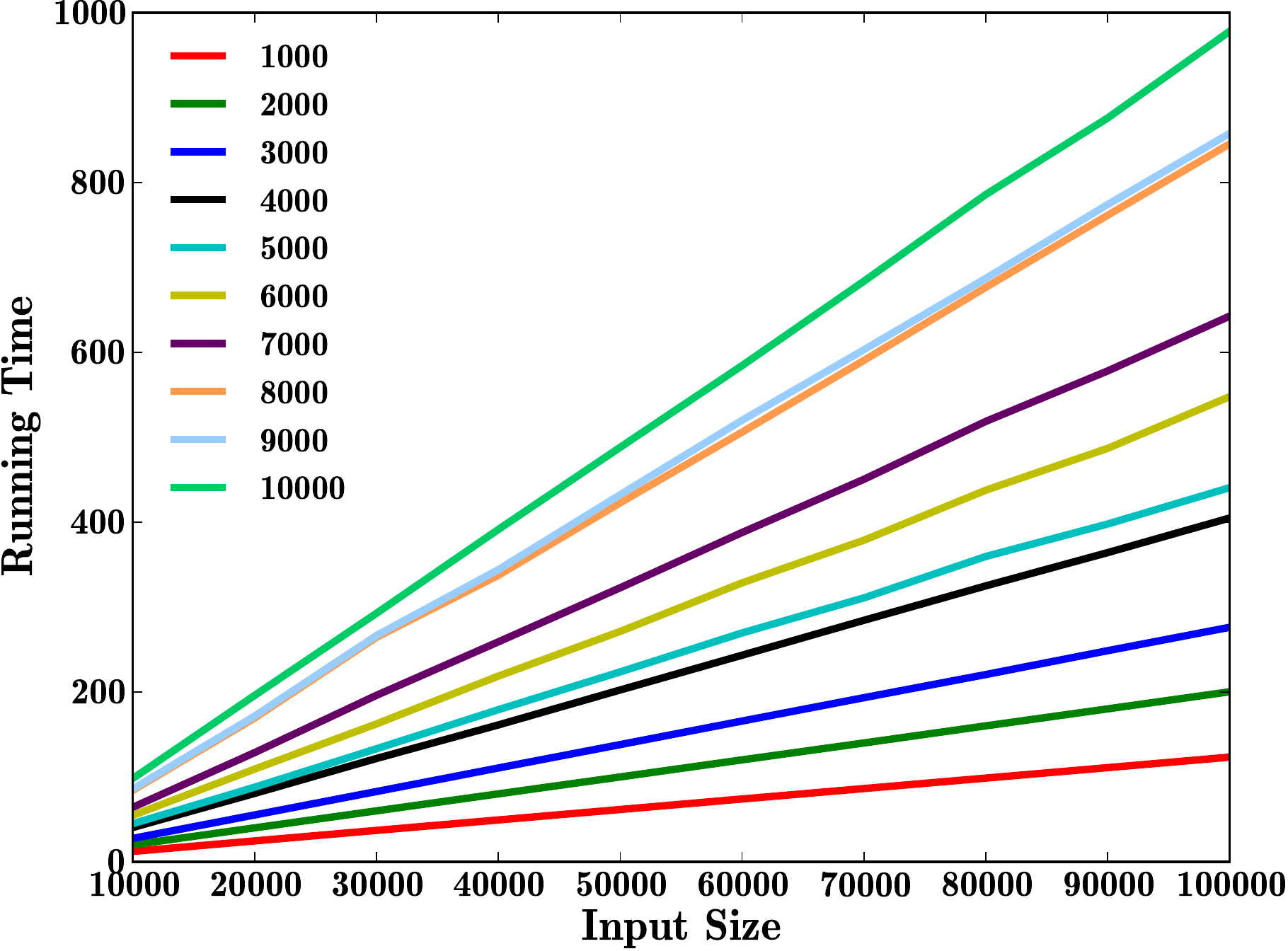}
\label{fig:time_vs_input}
}\vspace{2mm} 
\caption{Running time in seconds vs. input size.
Here, we measure only the running time of \FD (specifically Algorithm \ref{alg:freqDir2}). 
The sketch size is kept fixed at $\ell = 100$. The size of the synthetic input matrix is $n \times d$.
The value of $n$ is indicated on the x-axis. Different plot lines correspond to different values of $d$
(indicated in the legend). The running time is measured in seconds and is indicated on the y-axis.
It is clear from this plot that the running time of \FD is linear in both $n$ and $d$. Note
also that sketching a $10^5 \times 10^4$ dense matrix terminates in roughly 3 minutes.
}
\label{fig:time_vs_input}
\end{centering}
\end{figure}

\section*{Acknowledgments} 
The authors thank Qin Zhang for discussion on hardness results.  
We also thank Stanley Eisenstat, and Mark Tygert for their guidance regarding efficient $\svd$ rank-1 updates.  

\bibliographystyle{plain}
\bibliography{soda_journal}

\end{document}